\documentclass[letter,11pt]{article}
\usepackage{amsmath,amssymb,amsthm,color,graphicx,fullpage}
\usepackage{euscript}
\usepackage{times}
\newcommand{\seclab}[1]{\label{sec:#1}}

\newtheorem{theorem}{Theorem}[section]
\newtheorem{corollary}[theorem]{Corollary}
\newtheorem{lemma}[theorem]{Lemma}
\newtheorem{claim}[theorem]{Claim}
\graphicspath{{Figs/}}

\def\NN{\EuScript{N}}
\def \reals{{\mathbb R}}
\def \sphere{{\mathbb S}}

\def\dirtour{{\mathcal D}}

\def\bd{{\partial}}

\def\poly{\diamond}

\def\S{\mathcal{S}}

\def\T{\EuScript{T}}

\newcommand{\ignore}[1]{}

\def\inprod#1#2{\langle #1, #2\rangle}

\def\bisect{b}

\def\Nbrs{N}

\def\distfn{\varphi}

\def\G{{\sf G}}
\def\SDG{\mathop{\mathrm{SDG}}}

\def\DT{\mathop{\mathrm{DT}}}
\def\VD{\mathop{\mathrm{VD}}}
\def\Vor{\mathop{\mathrm{Vor}}}

\begin{document}

\begin{titlepage}

\title{Kinetic Stable Delaunay Graphs\thanks{%
A preliminary version of this paper appeared in {\it Proc. 26th Annual Symposium on Computational Geometry}, 2010, pp. 127--136.}}

\author{Pankaj K. Agarwal\thanks{%
    Department of Computer Science, Duke University, Durham, NC
    27708-0129, USA, {\tt pankaj@cs.duke.edu}.}
\and
Jie Gao\thanks{%
Department of Computer Science, Stony Brook University, Stony
Brook, NY 11794, USA, {\tt jgao@cs.sunysb.edu}. }
\and
Leonidas Guibas\thanks{%
Department of Computer Science, Stanford University, Stanford,
CA 94305, USA, {\tt guibas@cs.stanford.edu}.}
\and
Haim Kaplan\thanks{%
School of Computer Science, Tel Aviv University, Tel~Aviv 69978, Israel.
{\tt haimk@tau.ac.il}.}
\and
Vladlen Koltun\thanks{%
Department of Computer Science, Stanford University, Stanford,
CA 94305-9025, USA, {\tt vladlen@cs.stanford.edu}.}
\and
Natan Rubin\thanks{%
School of Computer Science, Tel Aviv University, Tel~Aviv 69978, Israel.
{\tt rubinnat@tau.ac.il}.}
\and
Micha Sharir\thanks{%
School of Computer Science, Tel Aviv University, Tel~Aviv 69978, Israel;
and Courant Institute of Mathematical Sciences, New York University,
New York, NY~~10012,~USA.  {\tt michas@tau.ac.il}.}
}

\maketitle

\begin{abstract}
We consider the problem of maintaining the Euclidean Delaunay
triangulation $\DT$ of a set $P$ of $n$ moving points in the plane, along algebraic tranjectories of constant description complexity.
Since the best known upper bound on the number of topological changes 
in the full Delaunay triangulation is only nearly cubic, we seek 
to maintain a suitable portion of the diagram that is less volatile 
yet retains many useful properties of the full triangulation.  
We introduce the notion of a {\em stable Delaunay graph}, which is 
a dynamic subgraph of the Delaunay triangulation. The stable Delaunay graph (a) is easy to
define, (b) experiences only a nearly quadratic number of discrete 
changes, (c) is robust under small changes of the norm, and (d)
possesses certain useful properties for further applications.

The stable Delaunay graph ($\SDG$ in short) is defined in terms of
a parameter $\alpha>0$, and consists of Delaunay edges $pq$ for
which the (equal) angles at which $p$ and $q$ see the corresponding
Voronoi edge $e_{pq}$ are at least $\alpha$.
We show that (i) $\SDG$ always contains at least roughly one third of the
Delaunay edges at any fixed time; (ii) it contains the
$\beta$-skeleton of $P$, for $\beta=1+\Omega(\alpha^2)$; (iii) it is
stable, in the sense that its edges survive for long periods of time,
as long as the orientations of the segments connecting (nearby) points
of $P$ do not change by much; and (iv) stable Delaunay edges remain
stable (with an appropriate redefinition of stability) if we
replace the Euclidean norm by any sufficiently close norm.

In particular, if we approximate the Euclidean norm by a polygonal 
norm (with a regular $k$-gon as its unit ball, with
$k=\Theta(1/\alpha)$), we can define and keep track of a Euclidean $\SDG$
by maintaining the full Delaunay triangulation of $P$ under the
polygonal norm (which is trivial to do, and which is known to involve only a
nearly quadratic number of discrete changes).

We describe two kinetic data structures for maintaining $\SDG$ when
the points of $P$ move along pseudo-algebraic trajectories of constant description complexity. The first
uses the polygonal norm approximation noted above, and the second is
slightly more involved, but significantly reduces the dependence of 
its performance on $\alpha$.  Both structures use $O^*(n)$ storage 
and process $O^*(n^2)$ events during the motion, each in $O^*(1)$ time.
(Here the $O^*(\cdot)$ notation hides multiplicative factors which 
are polynomial in $1/\alpha$ and polylogarithmic in $n$.)
\end{abstract}

\end{titlepage}

\section{Introduction}

\paragraph{Delaunay triangulations and Voronoi diagrams.} 
Let $P$ be a (finite) set of points in $\reals^2$. 
Let $\VD(P)$ and $\DT(P)$ denote the Voronoi diagram and Delaunay
triangulation of $P$, respectively. For a point $p \in P$, let
$\Vor(p)$ denote the Voronoi cell of $p$. 
The Delaunay triangulation $\DT=\DT(P)$ consists of all 
triangles whose circumcircles do not contain points of $P$ in their
interior. Its edges form the {\em Delaunay graph}, which is the 
straight-edge dual graph of the Voronoi diagram of $P$. That is,
$pq$ is an edge of the Delaunay graph if and only if
$\Vor(p)$ and $\Vor(q)$ share an edge, which we denote by $e_{pq}$.
This is equivalent to the existence of a circle passing through $p$
and $q$ that does not contain any point of $P$ in its interior---any
circle centered at a point on $e_{pq}$ and passing through $p$ and $q$
is such a circle. 
%
Delaunay triangulations and Voronoi diagrams are fundamental to much 
of computational geometry and its applications. 
See \cite{AK,Ed2} for a survey and a
textbook on these structures.

In many applications of Delaunay/Voronoi methods (e.g., mesh generation and kinetic collision detection) the points are moving continuously, so
these diagrams need to be efficiently updated as motion occurs.
Even though the motion of the nodes is continuous, the combinatorial and topological structure of the Voronoi and
Delaunay diagrams change only at
discrete times when certain critical events occur. Their evolution
under motion can be studied within the framework of {\em kinetic data
structures} (KDS in short) of Basch {\em et al.}~\cite{bgh-dsmd-99,285869,g-kdssar-98},
a general methodology for designing efficient algorithms for maintaining
such combinatorial attributes of mobile data. 

For the purpose of kinetic maintenance, Delaunay triangulations are 
nice structures, because, as mentioned above, they admit local 
certifications associated with individual triangles.  This makes 
it simple to maintain $\DT$ under point motion: an update is 
necessary only when one of these empty circumcircle conditions 
fails---this corresponds to cocircularities of certain subsets of
four points.\footnote{We assume that the motion of the points is sufficiently generic, so that no more than four points can become cocircular at any given time.} Whenever such an event happens, 
a single edge flip easily restores Delaunayhood. Estimating the 
number of such events, however, has been elusive---the problem 
of bounding the number of combinatorial changes in $\DT$ for 
points moving along semi-algebraic trajectories of constant description complexity has been in the 
computational geometry lore for many years; see \cite{TOPP}.

Let $n$ be the number of moving points in $P$. We
assume that each point moves along an algebraic trajectory of
fixed degree or, more generally, along pseudo-algebraic trajectory of constant description complexity (see Section~\ref{sec:Prelim} for a more formal
definition).
Guibas et al.~\cite{gmr-vdmpp-92} showed a roughly cubic upper bound of
$O(n^2 \lambda_s(n))$ on the number of discrete (also known as \textit{topological}) changes in $\DT$, where $\lambda_s(n)$ is the maximum length
of an $(n,s)$-Davenport-Schinzel sequence~\cite{SA95}, and $s$ is a constant
depending on the motions of the points. A substantial gap exists between this upper bound
and the best known quadratic lower bound~\cite{SA95}. 

It is thus desirable to find approaches for maintaining a substantial 
portion of $\DT$ that {\em provably} experiences only a nearly 
quadratic number of discrete changes, that is reasonably easy to define and 
maintain, and that retains useful properties for further applications.

\paragraph{Polygonal distance functions.} 
If the ``unit ball" of our
underlying norm is {\em polygonal} then things improve considerably.
In more detail, let $Q$ be a convex polygon with a constant
number, $k$, of edges.  It induces a {\em convex distance function}
$$d_Q(x,y) = \min\{\lambda \mid y\in x+\lambda Q\};$$ 
$d_Q$ is a metric if $Q$ is centrally symmetric with respect to the origin.

We can define the $Q$-Voronoi diagram
of a set $P$ of points in the plane in the usual way, as the
partitioning of the plane into Voronoi cells, so that the cell
$\Vor^\poly(p)$ of a point $p$ is
$\{ x\in\reals^2 \mid d_Q(x,p)=\min_{p'\in P}d_Q(x,p') \}$.
Assuming that the points of $P$ are in general position with respect
to $Q$, these cells are nonempty, have pairwise disjoint interiors,
and cover the plane. 

As in the Euclidean case, the $Q$-Voronoi diagram of $P$ has its 
dual representation, which we refer to as the {\em $Q$-Delaunay
triangulation} $\DT^\poly(P)=\DT^\poly$. A triple of points in $P$ define a
triangle in $\DT^\poly$ if and only if they lie on the boundary of some 
homothetic copy of $Q$ that does not contain any point of $P$ in its
interior. Assuming that $P$ is in general position, these $Q$-Delaunay
triangles form a triangulation of a certain simply-connected polygonal 
region that is contained in the convex hull of $P$.
Unlike the Euclidean case, it does not always coincide with the convex hull (see Figures~\ref{Fig:ConesCertif} and~\ref{Fig:AlmostTriangulation} for examples).
See Chew and Drysdale~\cite{CD} and Leven and Sharir~\cite{LS} for analysis of Voronoi and Delaunay diagrams of this kind.

For kinetic maintenance, polygonal Delaunay triangulations are
``better'' than Euclidean Delaunay triangulations because, as shown by
Chew~\cite{Chew}, when the points of $P$ move (in the algebraic
sense assumed above), the number of topological changes in the 
$Q$-Delaunay triangulation is only nearly quadratic in $n$.
One of
the major observations in this paper is that the \textit{stable portions} of the Euclidean Delaunay triangulation and the $Q$-Delaunay triangulation are closely related.

\paragraph{Stable Delaunay edges.} 
We introduce the notion of \textit{$\alpha$-stable Delaunay edges}, 
for a fixed parameter $\alpha>0$, defined as follows.  Let $pq$ be 
a Delaunay edge under the Euclidean norm, and let $\triangle pqr^+$ and $\triangle pqr^-$ 
be the two Delaunay triangles incident to $pq$.  Then $pq$ is 
called {\em $\alpha$-stable} if its opposite angles in these triangles
satisfy $\angle pr^+q + \angle pr^-q < \pi-\alpha$. (The case where
$pq$ lies on the convex hull of $P$ is treated as if one of $r^+,r^-$ 
lies at infinity, so that the corresponding angle $\angle pr^+q$ or 
$\angle pr^-q$ is equal to $0$.) An equivalent and more useful definition, in terms of the dual Voronoi diagram, is that 
$pq$ is $\alpha$-stable if the equal angles at which $p$ and $q$ 
see their common Voronoi edge $e_{pq}$ are at least $\alpha$.
See Figure \ref{Fig:LongDelaunay}.

\begin{figure}[htbp]
\begin{center}
\input{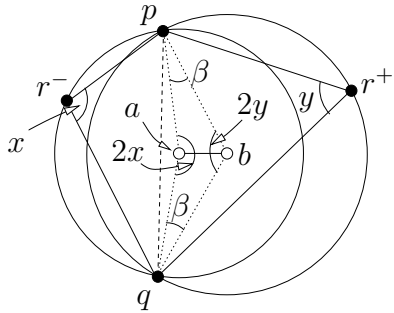}
\caption{\small \sf The points $p$ and $q$ see their common Voronoi edge $ab$ at (equal) angles $\beta$. This is equivalent to the angle condition $x+y=\pi-\beta$ for the two adjacent Delaunay triangles.}
\label{Fig:LongDelaunay}
\end{center}
\end{figure}

A justification for calling such edges stable lies in the following
observation: If a Delaunay edge $pq$ is $\alpha$-stable then it 
remains in $\DT$ during any continuous motion of $P$ for which
every angle $\angle prq$, for $r\in P\setminus\{p,q\}$, changes
by at most $\alpha/2$. This is clear because at the time $pq$ is
$\alpha$-stable we have
$\angle pr^+q + \angle pr^-q < \pi-\alpha$ for \textit{any} pair of points
$r^+$, $r^-$ lying on opposite sides of the line $\ell$ supporting $pq$, so,
if each of these angles change by at most $\alpha/2$ we still have
$\angle pr^+q + \angle pr^-q \le \pi$, which is easily seen to imply
that $pq$ remains an edge of $\DT$. (This argument also covers the cases when a point $r$ crosses $\ell$ from side to side: Since
each point, on either side of $\ell$, sees $pq$ at an angle of $\leq \pi-\alpha$, it follows that no point can cross
$pq$ itself -- the angle has to increase from $\pi-\alpha$ to $\pi$. Any other crossing of $\ell$ by a point $r$ causes $\angle prq$ to
decrease to $0$, and even if it increases to $\alpha/2$ on the other side of $\ell$, $pq$ is still an edge of $\DT$, as is easily checked.)
Hence, as long as the ``small angle change'' condition 
holds, stable Delaunay edges remain a ``long time'' in the
triangulation.

Informally speaking, the non-stable edges $pq$ of $\DT$ are those for $p$ and $q$
are almost cocircular with their two common Delaunay neighbors
$r^+$, $r^-$, and hence are more likely to get flipped ``soon". 

\paragraph{Overview of our results.}
Let $\alpha>0$ be a fixed parameter.
In this paper we show how to maintain a subgraph of the full Delaunay 
triangulation $\DT$, which we call a {\em $(c\alpha,\alpha)$-stable Delaunay graph} ($\SDG$ in short), so that (i) every edge of $\SDG$ is $\alpha$-stable, 
and (ii) every $c\alpha$-stable edge of $\DT$ belongs to $\SDG$, where $c>1$ is some (small) absolute constant.
Note that $\SDG$ is not uniquely defined, even when $c$ is fixed.

In Section \ref{sec:Prelim}, we introduce several useful definitions and show that the number of discrete changes in the $\SDG$s
that we consider
is nearly quadratic. 
What this analysis also implies is that if the true bound for
kinetic changes in a Delaunay triangulation is really close to cubic, then
the overhelming majority of these changes involve edges which never become stable and just flicker in and out of the diagram by cocircularity with their two Delaunay neighbors.

In Sections \ref{Sec:polygProp} and \ref{Sec:ReduceS} we show that $\SDG$ can be
maintained by a kinetic data structure that uses only near-linear
storage (in the terminology of \cite{bgh-dsmd-99}, it is {\em compact}), 
encounters only a nearly quadratic number of critical events 
(it is {\em efficient}), and processes each event in polylogarithmic 
time (it is {\em responsive}). For the second data structure, described in Section \ref{Sec:ReduceS}, can be slightly modified to ensure that each point appears at any time 
in only polylogarithmically many places in the structure (it then becomes 
{\em local}). 

The scheme described in Section \ref{Sec:polygProp} is based on a useful and interesting ``equivalence" connection
between the (Euclidean) $\SDG$ and a suitably defined ``stable" version of the Delaunay triangulation of $P$ under the ``polygonal" norm whose unit ball $Q$ is a regular
$k$-gon, for $k=\Theta(1/\alpha)$. As noted above, Voronoi and Delaunay structures under polygonal norms are particularly
favorable for kinetic maintenance because of Chew's
result~\cite{Chew}, showing that the number of topological changes in
$\DT^\poly(P)$ is $O^*(n^2k^4)$; here
the $O^*(\cdot)$ notation hides a factor that depends
sub-polynomially on both $n$ and $k$. In other words, the scheme simply maintains the ``polygonal" diagram $\DT^\poly(P)$ in its entirety, and selects from it those edges that are also stable edges of the Euclidean diagram $\DT$.

The major disadvantage of the solution in Section \ref{Sec:polygProp} is
the rather high (proportional to $\Theta(1/\alpha^4)$) dependence on 
$1/\alpha(\approx k)$ of the bound on the number of
topological changes. We do not know whether the upper bound $O^*(n^2k^4)$ on the number of topological changes in
$\DT^\poly(P)$ is nearly tight (in its dependence on $k$). 
To remedy this, we present in Section \ref{Sec:ReduceS} an 
alternative scheme for maintaining stable
(Euclidean) Delaunay edges. The scheme is reminiscent of the kinetic
schemes used in \cite{KineticNeighbors} for maintaining closest pairs
and nearest neighbors. It extracts $O^*(n)$ pairs of points of $P$ 
that are candidates for being stable Delaunay edges. Each point
$p\in P$ then runs $O(1/\alpha)$ \textit{kinetic and dynamic tournaments} involving
the other points in its candidate pairs. Roughly, these tournaments
correspond to shooting $O(1/\alpha)$ rays from $P$ in fixed directions and finding along each ray
the nearest point equally distant from $p$ and from some other
candidate point $q$. We show that $pq$ is a stable Delaunay edge if and
only if $q$ wins many (at least some constant number of) consecutive tournaments of $p$ (or $p$ wins many consecutive tournaments of $q$). A careful analysis shows that
the number of events that this scheme processes (and the overall 
processing time) is only $O^*(n^2/\alpha^2)$.

Section \ref{Sec:SDGProperties} establishes several useful properties of stable Delaunay graphs. In particular, we show that at
any given time the stable subgraph contains at least $\left[1-\frac{3}{2(\pi/\alpha-2)}\right]n$ Delaunay
edges, i.e., at least about one third of the maximum possible number of edges. In addition, we 
show that at any moment the $\SDG$ contains the closest pair, the so-called 
\textit{$\beta$-skeleton} of $P$, for $\beta=1+\Omega(\alpha^2)$ (see \cite{Crusts,Skeletons}), and the \textit{crust} of a sufficiently densely sampled point set along a smooth curve (see \cite{Amenta,Crusts}). 
We also extend the connection in Section \ref{Sec:polygProp} to arbitrary distance functions $d_Q$ whose unit ball $Q$ is sufficiently close (in the Hausdorff sense) to the Euclidean one (i.e., the unit disk).

\section{Preliminaries}\label{sec:Prelim}
\seclab{sdg}\seclab{ddg}
\paragraph{Stable edges in Voronoi diagrams.}
%
Let $\{u_0, \ldots, u_{k-1}\} \subset \sphere^1$ be a set of 
$k=\Theta(1/\alpha)$ equally spaced directions in $\reals^2$. For 
concreteness take $u_i = (\cos
(2\pi i/k), -\sin (2\pi i/k))$, $0\le i < k$ (so our directions $u_i$ go clockwise as $i$ increases).\footnote{%
The index arithmetic is modulo $k$, i.e., $u_i=u_{i+k}$.} 
For a point $p\in P$ and a unit vector
$u$ let $u[p]$ denote the ray $\{p+\lambda u\mid \lambda \geq 0\}$ that
emanates from $p$ in direction $u$.  For a pair of points $p,q \in P$ 
let $\bisect_{pq}$ denote the perpendicular bisector of $p$ and $q$.
If $\bisect_{pq}$ intersects $u_i[p]$, then the expression
\begin{equation}\label{Eq:DirectDist}
\distfn_i[p,q]=\frac{\|q-p\|^2}{2\inprod{q-p}{u_i}}
\end{equation}
is the distance
between $p$ and the intersection point of $\bisect_{pq}$ with
$u_i[p]$.
If $\bisect_{pq}$ does not
intersect $u_i[p]$ we define $\distfn_i[p,q] = \infty$.
 The point $q$ minimizes $\distfn_i[p,q']$, among all points
$q'$ for which $\bisect_{pq'}$ intersects $u_i[p]$, if and only if the
intersection between $\bisect_{pq}$ and $u_i[p]$ lies on the Voronoi edge $e_{pq}$. We call $q$ the {\em neighbor of $p$ in direction $u_i$}, 
and denote it by $\Nbrs_i(p)$; see Figure \ref{Fig:StableVoronoi}.

The {\em (angular) extent} of a Voronoi edge $e_{pq}$ of two points 
$p,q\in P$ is
the angle at which it is seen from either $p$ or $q$ (these two
angles are equal).  For a given angle $\alpha \le \pi$, 
$e_{pq}$  is  called {\em $\alpha$-long} (resp., {\em
$\alpha$-short}) if the extent of $e_{pq}$ is at least
 (resp., smaller than) $\alpha$. We also say that
$pq \in \DT(P)$ is {\em $\alpha$-long} (resp., {\em
$\alpha$-short}) if $e_{pq}$ is {\em $\alpha$-long} (resp., {\em
$\alpha$-short}). As noted in the introduction, these notions can also be defined (equivalently) in terms of the angles in the Delaunay triangulation: A Delaunay edge $pq$, which is not a hull edge, is $\alpha$-long if and only if $\angle pr^+q+\angle pr^-q\leq \pi-\alpha$,
where $\triangle pr^+q$ and $\triangle pr^-q$ are the two Delaunay triangles incident to $pq$. See Figure \ref{Fig:LongDelaunay}; hull edges are handled similarly, as discussed in the introduction.

Given parameters $\alpha'>\alpha>0$, we seek to construct (and
maintain under motion) an \emph{$(\alpha',\alpha)$-stable Delaunay
graph} (or \emph{stable Delaunay graph}, for brevity, which we further abbreviate as $\SDG$) of $P$, which
is any subgraph $\G$ of $\DT(P)$ with the following properties:
\begin{itemize}
\item[(S1)]
  Every $\alpha'$-long edge of $\DT(P)$ is an edge of
  $\G$.
\item[(S2)]
  Every edge of $\G$ is an $\alpha$-long edge of $\DT(P)$.
\end{itemize}
An $(\alpha',\alpha)$-stable Delaunay graph need not be
unique. In what follows, $\alpha'$ will always be some fixed (and reasonably small) multiple of $\alpha$.

\paragraph{Kinetic tournaments.} 
Kinetic tournaments were first studied by
Basch \textit{et al.}~\cite{bgh-dsmd-99}, for kinetically maintaining the lowest point in a set $P$ of $n$ points moving on some vertical line, say the $y$-axis, so that their trajectories are algebraic of bounded degree, as above. 
These tournaments are a key ingredient in the data structures that we will develop for maintaining stable Delaunay graphs. Such a tournament is 
represented and maintained using the following variant of a heap.
Let $T$ be a minimum-height balanced binary tree, with the points stored
at its leaves (in an arbitrary order). For an internal node $v\in T$,
let $P_v$ denote the set of points stored in the subtree rooted at $v$. At any
specific time $t$, each internal node $v$ stores the lowest point
among the points in $P_v$ at time $t$, which is called the {\em winner\/} at $v$.
The winner at the root is the desired overall lowest point of $P$.

To maintain $T$ we associate a certificate with each internal node $v$, which
asserts which of the two winners, at the left child and at the
right child of $v$, is the winner at $v$. This certificate remains
valid as long as (i) the winners at the children of $v$ do not change,
and (ii) the order along the $y$-axis between these two
``sub-winners'' does not change. The actual certificate caters
only to the second
condition; the first will be taken care of recursively.
Each certificate has an associated failure time, which is the next time
when these two winners switch their order along the $y$-axis.
We store all certificates in another heap, using the failure times
as keys.\footnote{Any ``standard'' heap that supports
  {\bf insert}, {\bf delete}, and {\bf deletemin} in $O(\log n)$
  time is good for our purpose.}
This heap of certificates is called the {\em event queue}.

Processing an event is simple. When the two sub-winners $p,q$ at some node $v$ change their order, we compute the new failure time of the certificate at $v$ (the first future time when $p$ and $q$ meet again), update the event queue accordingly, and propagate the new winner, say $p$, up the tree, revising the certificates at the ancestors of $v$, if needed.

If we assume that the  trajectories of each pair of points intersect at most $r$
times then
the overall number of changes of winners, and
therefore also the overall number of events, is at most
$\sum_v |P(v)| \beta_r(|P(v)|)= O(n \beta_r(n) \log n)$. Here
$\beta_r(n)=\lambda_r(n)/n$, and $\lambda_r(n)$ is the maximum length of a Davenport-Schinzel sequence
of order $r$ on $n$ symbols; see \cite{SA95}.

This is larger by a logarithmic factor than the maximum possible
number of times the lowest point along the $y$-axis can indeed change,
since this latter number is bounded by the complexity of the lower
envelope of the trajectories of the points in $P$ (which, as noted above, records the changes in the winner at the root of $T$).

Agarwal {\em et al.}~\cite{KineticNeighbors} show how to make
such a tournament also {\em dynamic\/}, supporting insertions and deletions of points. They replace the balanced binary tree $T$ by
 a {\em weight-balanced $(BB(\alpha))$ tree} \cite{NR73}
(and see also \cite{Mehlhorn}).  This allows us to insert a new point
anywhere we wish in $T$, and to delete any point from $T$,
in $O(\log n)$ time. Each such insertion or deletion may
change $O(\log n)$ certificates, along the corresponding search path,
and therefore updating the event queue takes $O(\log^2 n)$ time, including the time for the
structural updates of (rotations in) $T$; here $n$ denotes the
actual number of points in $T$, at the step where we perform
the insertion or deletion. The analysis of \cite{KineticNeighbors} is summarized in Theorem \ref{thm:kinetic-tour}.


\begin{theorem}[\textbf{Agarwal \textit{et al.}}~\cite{KineticNeighbors}] \label{thm:kinetic-tour}
A sequence of $m$ insertions and deletions into a kinetic tournament,
whose maximum size at any time is $n$ (assuming $m\ge n$), when
implemented as a weight-balanced tree in the manner described above,
generates at most $O(m\beta_{r+2}(n)\log n)$ events, with a total processing cost
of $O(m\beta_{r+2}(n)\log^2 n)$. Here $r$ is the maximum number of times a pair of points intersect, and $\beta_{r+2}(n)=\lambda_{r+2}(n)/n$.
Processing an update or a tournament event takes
$O(\log ^2 n)$ worst-case time. A dynamic kinetic tournament on $n$
elements can be constructed in $O(n)$ time.
\end{theorem}

\noindent {\it Remarks:} (1) Theorem \ref{thm:kinetic-tour} subsumes the static case too, by inserting all the elements ``at the beginning of time", and then tracing the kinetic changes. \\
\noindent (ii) Note that the amortized cost of an update or of processing a tournament event is only $O(\log n)$ (as opposed to the $O(\log^2n)$ worst-case cost).

\paragraph{Maintenance of an SDG.}
Let $P=\{p_1,\ldots,p_n\}$ be a set of points 
moving in $\reals^2$. Let $p_i(t)=(x_i(t),y_i(t))$ denote the position
of $p_i$ at time $t$. We call the motion of $P$
\emph{algebraic} if each $x_i(t),y_i(t)$ is a
polynomial function of $t$, and the \emph{degree} of motion of $P$ is the maximum
degree of these polynomials.
Throughout this paper we assume that the motion of $P$ is
algebraic and that its degree is bounded by a constant.
In this subsection we present a simple technique for maintaining a
$(2\alpha,\alpha)$-stable Delaunay graph. Unfortunately this
algorithm requires quadratic space.
It is based on the following easy observation (see Figure \ref{Fig:StableVoronoi}),
where $k$ is an integer, and the unit vectors (directions) $u_0,\ldots,u_{k-1}$ are as defined earlier.

\begin{lemma} \label{lem:alpha}
Let $\alpha = 2\pi/k$.
(i) If the extent of $e_{pq}$ is larger than $2\alpha$ then there are two consecutive directions 
$u_i$, $u_{i+1}$, such that $q$ is the neighbor of $p$ in directions $u_i$ and $u_{i+1}$. \\
(ii) If there are two consecutive directions $u_i,u_{i+1}$, such that $q$ is the neighbor of $p$ in both directions $u_i$ and $u_{i+1}$, then
the extent of $e_{pq}$ is at least $\alpha$.
\end{lemma}

\begin{figure}[htbp]
\begin{center}
\input{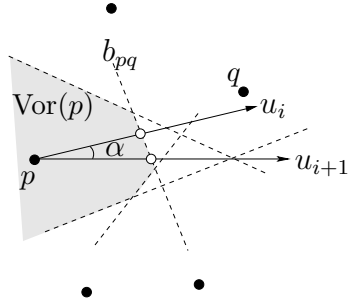}
\caption{\small \sf $q$ is the neighbor of $p$ in the directions $u_i$ and $u_{i+1}$, so the Voronoi edge $e_{pq}$ is $\alpha$-long.}
\label{Fig:StableVoronoi}
\end{center}
\end{figure}

The algorithm maintains Delaunay edges $pq$ such that there are two consecutive directions
$u_i$ and $u_{i+1}$ along which $q$ is the neighbor of $p$.
For each point $p$ and direction $u_i$ we get a set of at most $n-1$ piecewise
continuous functions of time, $\distfn_i[p,q]$, one for each point $q \not=
p$, as defined in (\ref{Eq:DirectDist}). (Recall that $\distfn_i[p,q]=\infty$ when $u_i[p]$
does not intersect $b_{pq}$.) By assumption on the motion of $P$,
for each $p$ and $q$, the domain in which $\distfn_i[p,q](t)$ is 
defined consists of a constant number of intervals.

For each point $p$, and ray $u_i[p]$, 
consider each function  $\distfn_i[p,q]$ as the
trajectory of a point moving
along the ray and corresponding to $q$. The algorithm maintains
these points in a dynamic and kinetic tournament $K_i(p)$ 
(see Theorem \ref{thm:kinetic-tour}) that keeps track of the 
minimum of $\{\distfn_i[p,q](t)\}_{q\neq p}$ over time.
For each pair of points $p$ and $q$ such that $q$
 wins in two consecutive tournaments, $K_i(p)$ and $K_{i+1}(p)$, of $p$,
 it keeps the edge $pq$ in
the stable Delaunay graph. It is trivial to update this graph as a by-product of the updates of the various tournaments.
The analysis of this
data structure is straightforward using Theorem \ref{thm:kinetic-tour},
and yields the following result.
\begin{theorem} \label{thm:ddj}
Let $P$ be a set of $n$ moving points in $\reals^2$ under algebraic
motion of bounded degree, let $k$ be an integer, and let $\alpha = 2\pi/k$.
A $(2\alpha,\alpha)$-stable Delaunay graph
of $P$ can be maintained using
$O(kn^2)$ storage and processing
 $O(kn^2\beta_{r+2}(n)\log n)$ events, for a total cost
of $O(kn^2\beta_{r+2}(n)\log^2 n)$ time.
The processing of each event takes
$O(\log ^2 n)$ worst-case time.
Here $r$ is a constant that depends on the degree of motion of $P$.
\end{theorem}

Later on, in Section \ref{Sec:ReduceS}, we will revise this approach and reduce the storage to nearly linear, by letting only 
a small number of points to participate in each tournament. The filtering procedure for the points makes the improved solution 
somewhat more involved.
\section{An SDG Based on Polygonal Voronoi Diagrams}
\label{sec:ViaPolygonal}
\label{Sec:polygProp}

Let $Q=Q_k$ be a regular $k$-gon
 for some even $k=2s$, circumscribed by the unit disk, and let $\alpha = \pi/s$ (this is the angle at which the center of $Q$ sees an edge). 
Let $\VD^\poly(P)$ and $\DT^\poly(P)$ denote the $Q$-Voronoi diagram and
the dual $Q$-Delaunay triangulation of $P$, respectively.
In this section we show that the set of
edges of $\VD^\poly(P)$ with sufficiently many \textit{breakpoints} (see below for details) form a
$(\beta,\beta')$-stable (Euclidean) Delaunay graph for appropriate multiples
$\beta,\beta'$ of $\alpha$.
Thus, by kinetically maintaining $\VD^\poly(P)$ (in its entirety),
we shall get ``for free'' a KDS for keeping track of a stable portion 
of the Euclidean DT.

\subsection{Properties of $\mathbf{VD^\poly(P)}$}
\label{Sec:PolygonalBackground}
We first review the properties of the (stationary)
$\VD^\poly(P)$ and $\DT^\poly(P)$. Then we consider the
kinetic version of these diagrams, as the points of $P$ move, and
review Chew's proof~\cite{Chew} that the number of topological
changes in these diagrams, over time, is only nearly quadratic
in $n$. Finally, we present a straightforward kinetic data structure 
for maintaining $\DT^\poly(P)$ under motion that uses linear storage, 
and that processes a nearly quadratic number of events, 
each in $O(\log n)$ time. 
Although later on we will take $Q$ to be a regular $k$-gon, the analysis in this subsection is more general, and we only assume here that $Q$ is an arbitrary convex $k$-gon, lying in general position with respect to $P$.

\paragraph{Stationary $Q$-diagrams.}
The {\em bisector} $b_{pq}^\poly$ between
two points $p$ and $q$, with respect to $d_Q(\cdot,\cdot)$, is the 
locus of all
placements of the center of any homothetic copy $Q'$ of $Q$ that
touches $p$ and $q$.
$Q'$ can be classified according to the pair of its edges, $e_1$ and $e_2$,
that touch $p$ and $q$, respectively. If we slide $Q'$ so that its
center moves along $b_{pq}^\poly$ (and its size expands or shrinks to
keep it touching $p$ and $q$), and the contact edges, $e_1$ and $e_2$,
remain fixed, the center traces a straight segment.
The bisector is a concatenation of $O(k)$ such segments. They
meet at {\em breakpoints}, which are placements of the center of a
copy $Q'$ that touches $p$ and $q$ and one of the contact points
is a vertex of $Q$; see Figure \ref{Fig:CornerContact}. We call such a 
placement a {\em corner contact} at the appropriate point. 
Note that a corner contact where some vertex $w$ of (a copy $Q'$ of) $Q$ touches $p$
has the property that the center of $Q'$ lies on the fixed ray
emanating from $p$ and parallel to the directed segment from $w$
to the center of $Q$.

\begin{figure}[htbp]
\begin{center}
\input{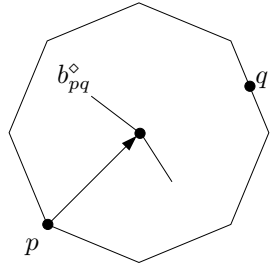}
\caption{\small \sf Each breakpoint on $\bisect_{pq}^\poly$ corresponds to a corner contact of $Q$ at one of the points $p,q$, so that $\partial Q$ also touches the other point.}\label{Fig:CornerContact}
\end{center}
\end{figure}

A useful property of bisectors and Delaunay edges, in the special case where $Q$ is a regular 
$k$-gon, which will be used in the next subsection, is that the breakpoints along a bisector 
$\bisect_{pq}^\poly$ alternate between corner contacts at $p$ and corner contacts at $q$. 
Indeed, assuming general position, each point $w\in\bd Q$ determines a unique
placement of $Q$ where it touches $p$ at $w$ and also touches $q$, as
is easily checked. A symmetric property holds when we interchange $p$
and $q$. Hence, as we slide the center of $Q$ along the bisector
$b^\poly_{pq}$, the points of contact of $\bd Q$ with $p$ and $q$ vary
continuously and monotonically along $\bd Q$. Consider two consecutive
corner contacts, $Q'$, $Q''$, of $Q$ at $p$ along $b^\poly_{pq}$, and
suppose to the contrary that the portion of $b^\poly_{pq}$ between
them is a straight segment, meaning that, within this portion, 
$\bd Q$ touches each of $p$, $q$ at a fixed edge. Since the center of
$Q$ moves along the angle bisector of the lines supporting these
edges (a property that is easily seen to hold for regular $k$-gons), it is easy to see that the distance between the two contact
points of $p$, at the beginning and the end of this sliding, and
the distance between the two contact points of $q$ (measured, say, 
on the boundary of the standard placement of $Q$) are equal. However,
this distance for $p$ is the length of a full edge of $\bd Q$, because
the motion starts and ends with $p$ touching a vertex, and therefore
the same holds for $q$, which is impossible (unless $q$ also starts
and ends at a vertex, which contradicts our general position
assumption).

Another well known property of $Q$-bisectors and Voronoi edges, for arbitrary convex polygons in general position with respect to $P$, is that
two bisectors $b^\poly_{pq_1}$, $b^\poly_{pq_2}$, can intersect at
most once (again, assuming general position), so every $Q$-Voronoi edge $e_{pq}^\poly$ is connected. 
Equivalently, this
asserts that there exists at most one homothetic placement of $Q$ at
which it touches $p$, $q_1$, and $q_2$. Indeed, since homothetic
placements of $Q$ behave like pseudo-disks (see, e.g., \cite{KLPS}),
the boundaries of two distinct homothetic placements of $Q$ intersect
in at most two points, or, in degenerate position, in at most two
connected segments. Clearly, in the former case the boundaries 
cannot both contain $p$, $q_1$, and $q_2$, and this also holds in the
latter case because of our general position assumption.

Consider next an edge $pq$ of $\DT^\poly(P)$. Its dual Voronoi edge
$e_{pq}^\poly$ is a portion of the bisector $b_{pq}^\poly$, and consists of those
center placements along $b_{pq}^\poly$ for which the corresponding copy $Q'$
has an {\em empty interior} (i.e., its interior is disjoint from $P$).
Following the notation of Chew~\cite{Chew}, we call $pq$ a
{\em corner edge} if $e_{pq}^\poly$ contains a breakpoint
(i.e., a placement with a corner contact); otherwise it is a
{\em non-corner edge}, and is therefore a straight segment.

\paragraph{Kinetic $Q$-diagrams.}
Consider next what happens to $\VD^\poly(P)$ and $\DT^\poly(P)$
as the points of $P$ move continuously with time.
In this case $\VD^\poly(P)$ changes
continuously, but undergoes topological
changes at certain critical times, called \emph{events}. There are 
two kinds of events:

\smallskip
\noindent (i) \textsc{Flip Event.}
A Voronoi edge $e_{pq}^\poly$ shrinks to a point, disappears, and is
``flipped'' into a newly emerging Voronoi edge $e_{p'q'}^\poly$.

\smallskip
\noindent (ii) \textsc{Corner Event.}
An endpoint of some Voronoi edge $e_{pq}^\poly$ becomes a breakpoint (a
corner placement). Immediately after this time $e_{pq}^\poly$ either 
gains a new straight segment, or loses such a segment, that it had before the event.

\smallskip
Some comments are in order: 

\smallskip
\noindent(a) A flip event
occurs when the four points $p,q,p',q'$ become ``cocircular'':
there is an empty homothetic copy $Q'$ of $Q$ that touches all four points.

\smallskip
\noindent(b) Only non-corner edges can participate in a flip event, as
both the vanishing edge $e_{pq}^\poly$ and the newly emerging edge
$e_{p'q'}^\poly$ do not have breakpoints near the event.

\smallskip
\noindent(c) A flip event entails a discrete change in the
Delaunay triangulation, whereas a corner event does not.
Still, for algorithmic purposes, we will keep track of both kinds of events.

\smallskip
We first bound the number of corner events.

\begin{lemma} \label{corners}
Let $P$ be a set of $n$ points in $\reals^2$ under algebraic
motion of bounded degree, and let $Q$ be a convex $k$-gon. 
The number of corner events in $\DT^\poly(P)$ is $O(k^2n\lambda_r(n))$, 
where $r$ is a constant that depends on the degree of motion
of $P$.
\end{lemma}

\begin{proof}
Fix a point $p$ and a vertex $w$ of $Q$, and consider all the corner
events in which $w$ touches $p$. As noted above, at any such event the
center $c$ of $Q$ lies on a ray $\gamma$ emanating from $p$ at a fixed
direction. (Since $p$ is moving, $\gamma$ is a moving ray, but its orientation remains fixed.) For each other point $q\in P\setminus\{p\}$, let $\distfn_\gamma^\poly[p,q]$
denote the distance, at time $t$, from $p$ along $\gamma$ to the center
of a copy of $Q$ that touches $p$ (at $w$) and $q$.
The value
$\min_q \distfn_\gamma^\poly[p,q](t)$ represents the intersection of
$\bd \Vor^\poly(p)$ with $\gamma$ at time $t$, where $\Vor^\poly(p)$ is the Voronoi cell of $p$ in $\VD^\poly(P)$. The point $q$ that attains the
minimum defines the Voronoi edge $e_{pq}^\poly$ (or vertex if the 
minimum is attained by more than one point $q$) of $\Vor^\poly(p)$ that $\gamma$ intersects.

In other words, we have a collection of $n-1$ partially defined
functions $\distfn_\gamma^\poly[p,q]$, and the breakpoints of their lower envelope
represent the corner events that involve the contact
of $w$ with $p$. By our assumption on the motion of $P$, each 
function $\distfn_\gamma^\poly[p,q]$ is piecewise algebraic, 
with $O(k)$ pieces. Each piece encodes a continuous contact of $q$ 
with a specific edge of $Q'$, and has constant description complexity. Hence (see, e.g., \cite[Corollary 1.6]{SA95}) the complexity of
the envelope is at most $O(k\lambda_r(n))$, for an appropriate constant
$r$. Repeating the analysis for each point $p$ and each vertex $w$ of $Q$, the lemma
follows.
\end{proof}

Consider next flip events. As noted, each flip event involves a 
placement of an empty homothetic copy $Q'$ of $Q$ that touches 
simultaneously four points $p_1,p_2,p_3,p_4$ of $P$, in this 
counterclockwise order along $\partial Q'$, so that the 
Voronoi edge $e_{p_1p_3}^\poly$,
which is a non-corner edge before the event, shrinks to a point
and is replaced by the non-corner edge $e_{p_2p_4}^\poly$. Let $e_i$
denote the edge of $Q'$ that touches $p_i$, for $i=1,2,3,4$. 

We fix the quadruple of edges $e_1,e_2,e_3,e_4$, bound the number
of flip events involving a quadruple contact with these edges,
and sum the bound over all $O(k^4)$ choices of four edges of $Q$.
For a fixed quadruple of edges $e_1,e_2,e_3,e_4$, we replace $Q$ by
the convex hull $Q_0$ of these edges, and note that any flip event
involving these edges is also a flip event
for $Q_0$. We therefore restrict our attention to $Q_0$, which is
a convex $k_0$-gon, for some $k_0\le 8$.

We note that if $(p,q)$ is a Delaunay edge
representing a contact of some homothetic copy $Q'_0$ of $Q_0$ where $p$ and $q$
touch two {\em adjacent} edges of $Q'_0$, then $(p,q)$ must be a corner
edge---shrinking $Q'_0$ towards the vertex common to the two edges,
while it continues to touch $p$ and $q$, will keep it empty, and
eventually reach a placement where either $p$ or $q$ touches a corner
of $Q'_0$.
The same (and actually simpler) argument applies to the case when $p$ and $q$ touch the same edge\footnote{In general position this does not occur, but it can happen at discrete time instances during the motion,} of $Q_0$.

\begin{figure}[htbp]
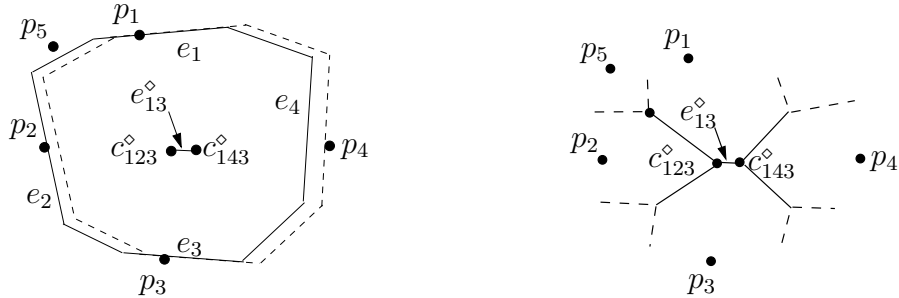

\begin{center}
\input{ChewProof.pstex_t}\hspace{3cm}\input{ChewProof1.pstex_t}
\caption{\small \sf Left: The edge $e_{13}^\poly$ in the diagram $\VD^\poly(P)$ before disappearing. The endpoint $c_{123}^\poly$ (resp., $c_{143}^\poly$) of $e_{13}^\poly$ corresponds to the homothetic copy of $Q_0$ whose edges $e_1,e_2,e_3$ (resp., $e_1,e_4,e_3$) are incident to the respective vertices $p_1,p_2,p_3$ (resp., $p_1,p_4,p_3$). Right: The tree of non-corner edges.}\label{Fig:ChewProof}
\end{center}
\end{figure}

Consider the situation just before the critical event takes place,
as depicted in Figure~\ref{Fig:ChewProof} (left).
The Voronoi edge $e_{p_1p_3}^\poly$ (to simplify the notation, we write this edge as $e_{13}^\poly$, and similarly for the other edges and vertices in this analysis) is delimited by two Voronoi vertices,
one, $c^\poly_{123}$, being the center of a copy of $Q_0$ which
touches $p_1,p_2,p_3$ at the respective edges $e_1,e_2,e_3$,
and the other, $c_{143}^\poly$, being the center of a copy of
$Q_0$ which touches $p_1,p_4,p_3$ at the respective edges
$e_1,e_4,e_3$. Consider the two other Voronoi edges $e_{12}^\poly$ and
$e_{23}^\poly$ adjacent to $c_{123}^\poly$, and
the two Voronoi edges $e_{14}^\poly$ and
$e_{43}^\poly$ adjacent to $c_{143}^\poly$.
Among them, consider only those which are non-corner edges;
assume for simplicity that they all are.
For specificity, consider the edge $e_{12}^\poly$. As we move
the center of $Q_0$ along that edge away from $c_{123}^\poly$,
$Q_0$ loses the contact with $p_3$; it shrinks on the side of
$p_1p_2$ which contains $p_3$ (and $p_4$, already away from $Q_0$),
and expands on the other side. Since this is a non-corner edge,
its other endpoint is a placement where the (artificial) edge
$e_{12}$ of $Q_0$ between $e_1$ and $e_2$ touches another point
$p_{5}$. Now, however, since $e_{12}$ is adjacent to both edges
$e_1$, $e_2$, the new Voronoi edges $e^\poly_{15}$ and
$e_{25}^\poly$ are both corner edges.

Repeating this analysis to each of the other three Voronoi edges
adjacent to $e_{13}^\poly$, we get a tree of non-corner Voronoi edges,
consisting of at most five edges, so that all the other Voronoi edges
adjacent to its edges are corner edges. As long as no discrete change occurs at any of the surrounding corner edges, the tree can undergo only $O(1)$ discrete changes, because all its edges are defined by a total of $O(1)$ points of $P$. When a corner edge undergoes a discrete change, this can affect only $O(1)$ adjacent non-corner trees of the above kind. Hence, the number of changes in non-corner edges is proportional to the number
of changes in corner edges, which, by Lemma \ref{corners} (applied to $Q_0$) is $O(n\lambda_r(n))$. Multiplying by the $O(k^4)$ choices of quadruples of edges of $Q$, we thus obtain:

\begin{theorem} \label{Thm:PolygonalVoronoi}
Let $P$ be a set of $n$ moving points in $\reals^2$ under algebraic
motion of bounded degree, and let $Q$ be a convex $k$-gon.
The number of topological changes in $\VD^\poly(P)$ with respect to $Q$
 is $O(k^4n\lambda_r(n))$, where $r$ is a
constant that depends on the degree of motion of $P$.
\end{theorem}

\paragraph{Kinetic maintenance of $\mathbf{VD^\poly(P)}$ and 
$\mathbf{DT^\poly(P)}$.}
As already mentioned, it is a fairly trivial task to maintain 
$\DT^\poly(P)$ and $\VD^\poly(P)$ kinetically, as the points of $P$ 
move. All we need to do is to assert the correctness of the present 
triangulation by a collection of local certificates, one for each edge 
of the diagram, where the certificate of an edge asserts that the two 
homothetic placements $Q^-,Q^+$ of $Q$ that circumscribe the two 
respective adjacent $Q$-Delaunay triangles 
$\triangle pqr^-,\triangle pqr^+$, are such that $Q^-$ does not 
contain $r^+$ and $Q^+$ does not contain $r^-$. The failure time of this certificate is the first time (if one exists) at which $p,q,r^-$, and $r^+$ become $Q$-cocircular---they all lie on the boundary of a common homothetic copy of $Q$. Such an event corresponds to a flip event in $\DT^\poly(P)$. If $pq$ is an edge of the periphery of $\DT^\poly(P)$, so that $\triangle pqr^+$ exists but $\triangle pqr^-$ does not, then $Q^-$ is a limiting wedge bounded by rays supporting two {\it consecutive} edges of (a copy of) $Q$, one passing through $p$ and one through $q$ (see Figure \ref{Fig:ConesCertif}).
The failure time of the corresponding certificate is the first time (if any) at which $r^+$ also lies on the boundary of that wedge.

We maintain the breakpoints using ``sub-certificates", each of which asserts that $Q^-$, say, touches each of $p,q,r^-$ at respective specific edges (and similarly for $Q^+$). The failure time of this sub-certificate is the first failure time when one of $p,q$ or $r^-$ touches $Q^-$ at a vertex. In this case we have a corner event---two of the adjacent Voronoi edges terminate at a corner placement. Note that the failure time of each sub-certificate can be computed in $O(1)$ time. Moreover, for a fixed collection of valid sub-certificates, the failure time of an initial certificate (asserting non-cocircularity) can also be computed in $O(1)$ time (provided that it fails before the failures of the corresponding sub-certificates), because we know the four edges of $Q^-$ involved in the contacts.

\begin{figure}[htbp]
\begin{center}
\input{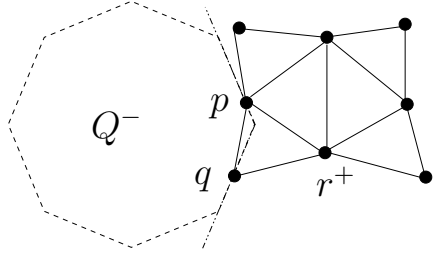}
\caption{\small \sf If $r^-$ does not exist then $Q^-$ is a limiting wedge bounded by rays supporting two consecutive edges of (a copy of) $Q$.}
\label{Fig:ConesCertif}
\end{center}
\end{figure}

We therefore maintain an event queue that stores and updates all the 
active failure times (there are only $O(n)$ of them at any given time---the bound is 
independent of $k$, because they correspond to actual $\DT$ edges. When a sub-certificate fails we do not change 
$\DT^\poly(P)$, but only update the corresponding Voronoi edge, by 
adding or removing a segment and a breakpoint, and by replacing the 
sub-certificate by a new one; we also update the cocircularity certificate 
associated with the edge, because one of the contact edges has changed.
When a cocircularity certificate fails we update $\DT^\poly(P)$ and 
construct $O(1)$ new sub-certificates and certificates. Altogether, each update of the diagram takes $O(\log n)$ time. We thus have

\begin{theorem}\label{Thm:MaintainPolygDT}
Let $P$ be a set of $n$ moving points in $\reals^2$ under algebraic
motion of bounded degree, and let $Q$ be a convex $k$-gon.
$\DT^\poly(P)$ and $\VD^\poly(P)$ can be maintained using
$O(n)$ storage and $O(\log n)$ update time, so that
$O(k^4n\lambda_r(n))$ events are processed, where $r$ is  a
constant that depends on the degree of motion of $P$.
\end{theorem}

\subsection{Stable Delaunay edges in $\mathbf{DT^\poly(P)}$}
We now restrict $Q$ to be a regular $k$-gon.
Let $v_0,\ldots,v_{k-1}$ be
the vertices of $Q$ arranged in a clockwise direction, with $v_0$ the leftmost.   
We call a homothetic copy of $Q$ whose vertex $v_j$ touches a point $p$, a
{\em $v_j$-placement of $Q$ at $p$}. Let
$u_j$ be the direction of the vector that connects $v_j$ with the 
center of $Q$, for each $0\leq j< k$ (as in Section \ref{sec:Prelim}). See Figure \ref{Fig:Placement} (left).

We follow the machinery in the proof of Lemma~\ref{corners}. That is, for any pair $p,q\in
P$ let $\varphi^\poly_j[p,q]$ denote the distance from $p$ to the point
$u_j[p]\cap \bisect^\poly_{pq}$; we put $\varphi^\poly_j[p,q]=\infty$ if
$u_j[p]$ does not intersect $\bisect^\poly_{pq}$. If
$\varphi^\poly_j[p,q]<\infty$ then the point $\bisect^\poly_{pq}\cap u_j[p]$ is
the center of the $v_j$-placement $Q'$ of $Q$ at $p$
that also touches $q$, and it is easy to see that there is a unique such point.
The value $\varphi^\poly_j[q,p]$ is equal to the circumradius
of $Q'$. See Figure \ref{Fig:Placement} (middle).

\begin{figure}[htbp]
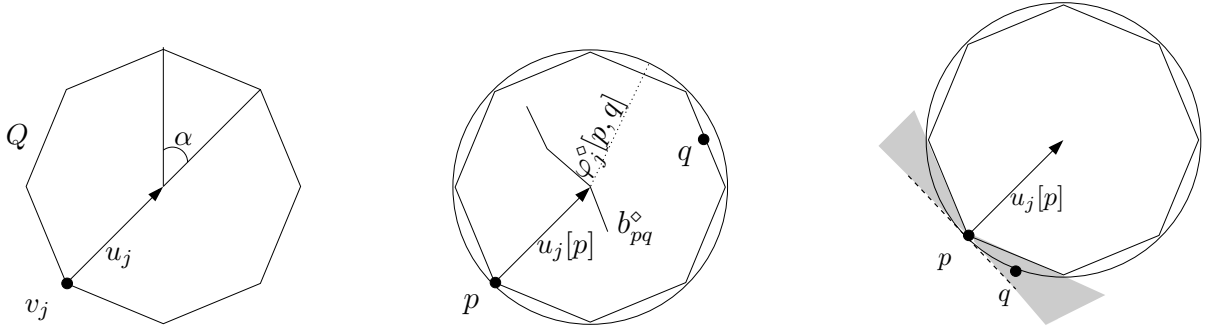

\begin{center}
\input{Regular.pstex_t}\hspace{2cm}\input{PolygonalDist.pstex_t}\hspace{2cm}\input{UndefinedPolygDist.pstex_t}
\caption{\small \sf Left: \sf $u_j$ is the direction of the vector connecting vertex $v_j$ to the center of $Q$. Middle:
The function $\varphi_j^\poly[p,q]$ is equal to the radius of
the circle that circumscribes the $v_j$-placement of $Q$ at
$p$ that also touches $q$.
Right: The case when $\varphi^\poly_j[p,q]=\infty$ while
  $\varphi_j[p,q]<\infty$. In this case $q$ must lie in one of
the shaded wedges.
}\label{Fig:Placement}
\end{center}
\end{figure}

The \textit{neighbor} $\Nbrs^\poly_j[p]$ of $p$ in direction $u_j$ is defined to be the point $q\in
P\setminus\{p\}$ that minimizes $\varphi^\poly_j[p,q]$. Clearly, for any
$p,q\in P$ we have $\Nbrs^\poly_j[p]=q$ if and only if there is an empty
$v_j$-placement $Q'$ of $Q$ at $p$ so that $q$ touches one of
its edges.


\smallskip
\noindent{\bf Remark:}
Note that, in the Euclidean case, we have $\varphi_j[p,q]<\infty$ if
and only if the angle between $\overline{pq}$ and $u_j[p]$ is at most
$\pi/2$.  In contrast, $\varphi^\poly_j[p,q]<\infty$ if and only if
the angle between $\overline{pq}$ and $u_j[p]$ is at most
$\pi/2-\pi/k=\pi/2-\alpha/2$. Moreover, we have $\varphi_j[p,q]\leq
\varphi^\poly_j[p,q]$. Therefore, $\varphi^\poly_j[p,q]<\infty$ always
implies $\varphi_j[p,q]<\infty$, but not vice versa; see Figure
\ref{Fig:Placement} (right). Note also that in both the Euclidean and the polygonal cases, the respective quantities $N_j[p]$ and $N_j^\poly[p]$ may be undefined.

\begin{lemma}\label{Thm:LongEucPoly}
Let $p,q\in P$ be a pair of points such that $\Nbrs_j(p)=q$ for 
$h\geq 3$ consecutive indices, say $0\leq j\leq h-1$.
Then for each of these indices, except possibly for the first and the last one, we also have $\Nbrs^\poly_j[p]=q$.
\end{lemma}

\begin{proof}
Let $w_1$ (resp., $w_2$) be the point at which the ray $u_0[p]$ (resp., $u_{h-1}[p]$) hits the edge $e_{pq}$ in $\VD(P)$. (By assumption, both points exist.)
Let $D_1$ and $D_2$ be the disks centered at $w_1$ and $w_2$, respectively, and touching $p$ and $q$. By definition, neither of these disks contains a point of $P$ in its interior. The angle between the tangents to $D_1$ and $D_2$ at $p$ or at $q$ (these angles are equal) is $\beta=(h-1)\alpha$; see Figure \ref{Fig:ProvePolyg} (left).

\begin{figure}[htbp]
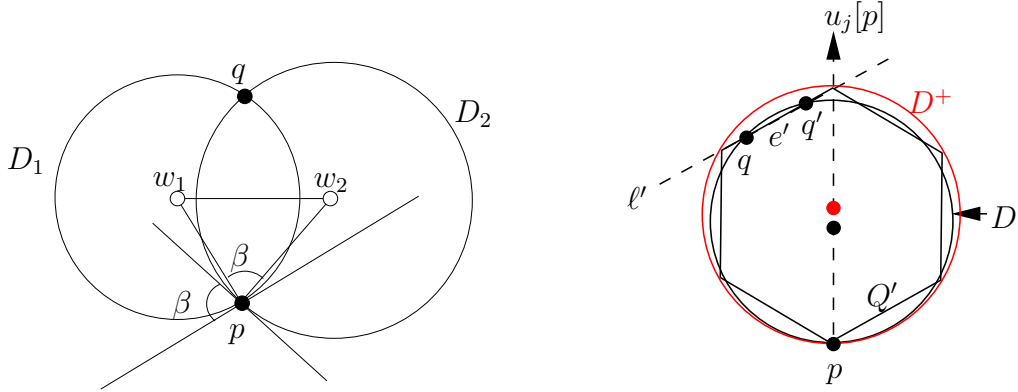

\begin{center}
\input{ProvePolyg1.pstex_t}\hspace{2cm}\input{ProvePolyg2.pstex_t}
\caption{\small \sf Left: The angle between the tangents to $D_1$ and $D_2$ at $p$ (or at $q$) is equal to
  $\angle w_1pw_2= \beta=(h-1)\alpha$. Right: The line $\ell'$ crosses $D$ in a chord $qq'$ which is fully
  contained in $e'$.}\label{Fig:ProvePolyg}
\end{center}
\end{figure}

Fix an arbitrary index $1\leq j\leq h-2$, so $u_j[p]$ intersects $e_{pq}$ 
and forms an angle of at least $\alpha$ with each of ${pw}_1,{pw}_2$.
Let $Q'$ be the $v_j$-placement of $Q$
at $p$ that touches $q$. To see that such a placement exists, we note that, by the preceding remark, it suffices to show that the angle between
$\overline{pq}$ and $u_j[p]$ is at most $\pi/2-\alpha/2$; that is, to rule out the case where $q$ lies in one of the shaded wedges in Figure \ref{Fig:Placement} (right). This case is indeed impossible, because then one of $u_{j-1}[p],u_{j+1}[p]$ would form an angle greater than $\pi/2$ with 
$\overline{pq}$, contradicting the assumption that both of these rays intersect the (Euclidean) $\bisect_{pq}$.

We claim that $Q' \subset D_1\cup D_2$. 
Establishing this property for every $1\leq j\leq h-2$ will complete 
the proof of the lemma. 
Let $e'$ be the edge of $Q'$ passing through $q$. See Figure \ref{Fig:ProvePolyg} (right). Let $D$ be the disk 
whose center lies on $u_j[p]$ and which passes through $p$ and $q$, 
and let $D^+$ be the circumscribing disk of $Q'$. 
Since $q\in \partial D$ and is interior to $D^+$, and since $D$ and 
$D^+$ are centered on the same ray $u_j[q]$ and pass through $p$, it 
follows that $D \subset D^+$. 
The line $\ell'$ containing $e'$ crosses $D$ in a chord $qq'$ that
is fully contained in $e'$. The angle between the tangent to $D$ at 
$q$ and the chord $qq'$ is equal to the angle at which $p$ sees $qq'$. 
This is smaller than the angle at which $p$ sees $e'$, which in turn 
is equal to $\alpha/2$.

Arguing as in the analysis of $D_1$ and $D_2$, the tangent to $D$ at $q$ forms an angle of at least $\alpha$ with each of the tangents to $D_1,D_2$ at $q$, and hence $e'$ forms an angle of at least $\alpha/2$ with each of these tangents; see Figure \ref{Fig:ProvePolyg2} (left).
The line $\ell'$ marks two chords $q_1q,qq_2$ within the respective disks $D_1,D_2$. We claim that $e'$ is fully contained in their union $q_1q_2$. Indeed, the angle $q_1pq$ is equal to the angle between $\ell'$ and the tangent to $D_1$ at $q$, so $\angle q_1pq\geq \alpha/2$.
On the other hand, the angle at which $p$ sees $e'$ is $\alpha/2$, which is smaller. This, and the symmetic argument involving $D_2$, are easily seen to imply the claim.

\begin{figure}[htbp]
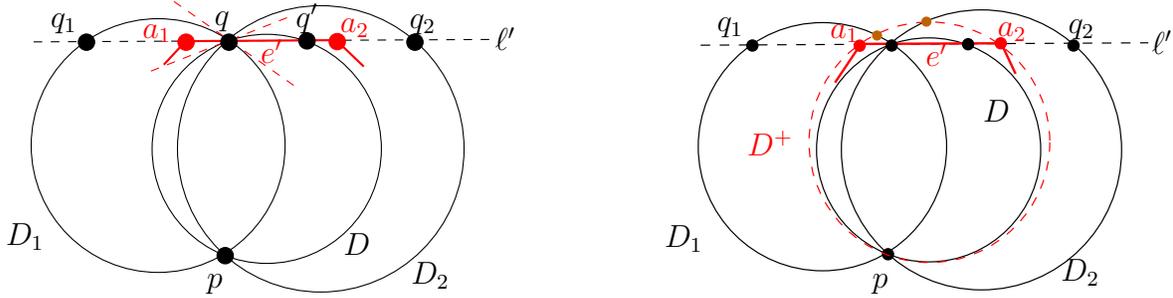

\begin{center}
\input{ProvePolyg3.pstex_t} \hspace{2cm} \input{ProvePolyg4.pstex_t}
\caption{\small \sf Left: The line $\ell'$ forms an angle of at least $\alpha/2$ with each of the tangents to $D_1,D_2$ at $q$. Right: The edge $e'=a_1a_2$ of $Q'$ is fully contained in $D_1\cup D_2$.}\label{Fig:ProvePolyg2}
\end{center}
\end{figure}

Now consider the circumscribing disk $D^+$ of $Q'$. Denote the endpoints of $e'$ as $a_1$ and $a_2$, where $a_1$ lies in $q_1q$ and $a_2$ lies in $qq_2$.
Since the ray $\overline{pa}_1$ hits $\partial D^+$ before hitting $D_1$, and the ray $\overline{pq}$ hits these circles in the reverse order, it follows that the second intersection of $\partial D_1$ and $\partial D^+$ (other than $p$) must lie on a ray from $p$ which lies between the rays $\overline{pa}_1,\overline{pq}$ and thus crosses $e'$. See Figure \ref{Fig:ProvePolyg2} (right).
Symmetrically, the second intersection point of $\partial D_2$ and $\partial D^+$ also lies on a ray which crosses $e'$.

It follows that the arc of $\partial D^+$ delimited by these intersections and containing $p$ is fully contained in $D_1\cup D_2$.
Hence all the vertices of $Q'$ (which lie on this arc) lie in $D_1\cup D_2$. This, combined with the argument in the preceding paragraphs, is easily seen to imply that $Q'\subseteq D_1\cup D_2$, so its interior does not contain points of $P$, which in turn implies that $\Nbrs_j^\poly[p]=q$. 
As noted, this completes the proof of the lemma.
\end{proof}

Since $Q$-Voronoi edges are connected, Lemma~\ref{Thm:LongEucPoly} implies that $e_{pq}^\poly$ is ``long", in the sense that it contains at least $h-2$ breakpoints that represent corner placements at $p$, interleaved (as promised in Section \ref{Sec:PolygonalBackground}) with at least $h-3$ corner placements at $q$.
This property is easily seen to hold also under the weaker assumptions that: (i) for the first and the last indices $j=0,h-1$, the point $\Nbrs_j[p]$ either is equal to $q$ or is undefined, and (ii) for the rest of the indices $j$ we have $\Nbrs_j[p]=q$ and $\varphi^\poly_j[p,q]<\infty$ (i.e., the $v_j$-placement of $Q$ at $p$ that touches $q$ exists).
In this relaxed setting, it is now possible that any of the two points $w_1,w_2$ lies at infinity, in which case the corresponding disk $D_1$ or $D_2$ degenerates into a halfplane. This stronger version of 
Lemma~\ref{Thm:LongEucPoly} is used in the proof of the converse 
Lemma~\ref{Thm:LongPolygEuc}, asserting that every edge $e^\poly_{pq}$ in $\VD^\poly(P)$ with sufficiently many breakpoints has a stable counterpart $e_{pq}$ in $\VD(P)$.

\begin{lemma}\label{Thm:LongPolygEuc}
Let $p,q\in P$ be a pair of points such that $\Nbrs_j^\poly[p]=q$ for at least three consecutive indices $j\in \{0,\ldots, k-1\}$.
Then for each of these indices, except possibly for the first and the last one, we have $\Nbrs_j[p]=q$.
\end{lemma}
\begin{proof} 
Again, assume with no
  loss of generality that $\Nbrs_j^\poly[p]=q$ for $0\leq j\leq h-1$, with
  $h\geq 3$.  Suppose to the contrary that, for some $1\leq j\leq h-2$, we have $\Nbrs_j[p]\neq q$.  Since $\Nbrs^\poly_j[p]=q$
  by assumption, we have $\varphi_j[p,q]\leq
  \varphi_j^\poly[p,q]<\infty$, so there exists $r\in P$ for which
  $\varphi_j[p,r]<\varphi_j[p,q]$.  Assume with no loss of generality
  that $r$ lies to the left of the line from $p$ to $q$. In this case $\varphi_{j-1}[p,r]<\varphi_{j-1}[p,q]<\infty$. Indeed, we have (i) $\Nbrs_{j-1}^\poly[p]=q$ by assumption,
  so $\varphi_{j-1}^\poly[p,q]<\infty$, and (ii) $\varphi_{j-1}[p,q]\leq
  \varphi_{j-1}^\poly[p,q]$. Moreover,
  because $r$ lies to the left of the line from $p$ to $q$, the orientation of $\bisect_{pr}$ lies counterclockwise to that of $\bisect_{pq}$,
  implying that $\varphi_{j-1}[p,q]<\infty$.
  See Figure \ref{Fig:Converse}. Since $u_j[p]$ hits $\bisect_{pr}$ before hitting $\bisect_{pq}$, any ray emanating from $p$ counterlockwise to $u_j[p]$ must do the same, so we have $\varphi_{j-1}[p,r]<\varphi_{j-1}[p,q]$, as claimed. Similarly, we get that either
  $\varphi_{j-2}[p,r]<\varphi_{j-2}[p,q]<\infty$ or
  $\varphi_{j-2}[p,r]\leq \varphi_{j-2}[p,q]=\infty$ (where the latter
  can occur only for $j=1$).  Now applying (the extended version of)
  Lemma~\ref{Thm:LongEucPoly} to the point set $\{p,q,r\}$ and to
  the index set $\{j-2,j-1,j\}$, we get that
  $\varphi^\poly_{j-1}[p,r]<\varphi^\poly_{j-1}[p,q]$. But this
  contradicts the fact that $\Nbrs_{j-1}^\poly[p]=q$.
\end{proof}

\begin{figure}[htbp]
\begin{center}
\input{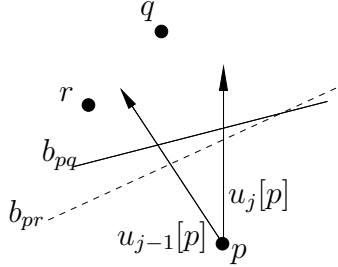}
\caption{\small \sf Proof of Lemma~\ref{Thm:LongPolygEuc}. If $\Nbrs_j[p]\neq q$ because some $r$, lying to the left of the line from $p$ to $r$, satisfies $\varphi_{j}[p,r]<\varphi_{j}[p,q]$. Since $\varphi_{j-1}[p,q]<\varphi_{j-1}^\poly[p,q]<\infty$, we have $\varphi_{j-1}[p,r]<\varphi_{j-1}[p,q]$.}
\label{Fig:Converse}
\end{center}
\end{figure}

\paragraph{Maintaining an SDG using $\mathbf{VD^\poly(P)}$.}
Lemmas~\ref{Thm:LongEucPoly} 
and~\ref{Thm:LongPolygEuc} together imply that an $\SDG$ can be maintained
using the fairly straightforward kinetic algorithm for maintaining 
the whole $\VD^\poly(P)$, provided by Theorem~\ref{Thm:MaintainPolygDT}. 
We use 
$\VD^\poly(P)$ to maintain the graph $\G$ on $P$, whose edges are all 
the pairs $(p,q)\in P\times P$ such that $p$ and $q$ define an edge 
$e^\poly_{pq}$ in $\VD^\poly(P)$ that contains at least seven 
breakpoints. As shown in Theorem \ref{Thm:MaintainPolygDT}, this can 
be done with $O(n)$ storage, $O(\log n)$ update time, and 
$O(k^4n\lambda_r(n))$ updates (for an appropriate $r$).  We claim that $\G$ is a 
$(6\alpha,\alpha)$-$\SDG$ in the Euclidean norm. 

Indeed, if two points $p,q\in P$ define a $6\alpha$-long edge $e_{pq}$ in $\VD(P)$ then
this edge stabs at least six rays $u_j[p]$ emanating from $p$, and at least six rays $u_j[q]$ emanating from $q$.
Thus, according to Lemma~\ref{Thm:LongEucPoly}, $\VD^\poly(P)$ contains the edge $e_{pq}^\poly$ with at least four breakpoints corresponding to corner placements of $Q$ at $p$ that touch $q$, and at least four breakpoints corresponding to corner placements of $Q$ at $q$ that touch $p$. Therefore, $e_{pq}^\poly$ contains at least $8$ breakpoints, so $(p,q)\in \G$.

For the second part, if $p,q\in P$ define an edge $e_{pq}^\poly$ in $\VD^\poly(P)$ with at least $7$ breakpoints then, by the interleaving property of breakpoints, we may assume, without loss of generality, that at least four of these breakpoints correspond to $P$-empty corner placements of $Q$ at $p$ that touch $q$.
Thus, Lemma~\ref{Thm:LongPolygEuc} implies that $\VD(P)$ contains the edge $e_{pq}$, and that this edge is hit by at least two consecutive rays $u_j[p]$.
But then, as observed in Lemma \ref{lem:alpha}, the edge $e_{pq}$ is $\alpha$-long in $\VD(P)$.
We thus obtain the main result of this section.
 
\begin{theorem}\label{Thm:MaintainSDGPolyg}
Let $P$ be a set of $n$ moving points in $\reals^2$ under algebraic
motion of bounded degree, 
and let $\alpha \ge 0$ be a parameter.  A
$(6\alpha,\alpha)$-stable Delaunay graph of $P$ can be maintained by a 
KDS of linear size that processes
 $O(n\lambda_r(n)/\alpha^4)$ events, where $r$ is a
constant that depends on the degree of motion of $P$, and that 
updates the SDG  at each event in $O(\log n)$ time.
\end{theorem}

\section{An Improved Data Structure}
\label{Sec:ReduceS}

The data structure of Theorem~\ref{Thm:MaintainSDGPolyg} requires
$O(n)$ storage but the best bound we have on the number of
events it may encounter is $O^*(n^2/\alpha^4)$, which is 
much larger than the number of events encountered by the data structure
of Theorem~\ref{thm:ddj} (which, in terms of the dependence on $\alpha$, is only $O^*(n^2/\alpha)$).  In this
section we present an alternative data structure that requires
$O^*(n/\alpha^2)$ space and $O^*(n^2/\alpha^2)$ overall processing
time.  The structure processes each event in $O^*(1/\alpha)$ time and is also {\it local}, in the sense that each point is stored at only $O^*((1/\alpha)^2)$ places in the structure. 

\paragraph{Notation.}
We use the directions $u_i$ and the associated quantities $\Nbrs_i[p]$ and  $\varphi_i[p,q]$ defined in
Section \ref{sec:Prelim}. We assume that $k$, the number of canonical directions, is even, and write, as in Section \ref{sec:Prelim}, $k=2s$.
We denote by $C_i$ the cone (or wedge) with apex at the origin that is bounded by
$u_i$ and $u_{i+1}$. Note that $C_i$ and $C_{i\pm s}$ are antipodal.
As before, for a vector $u$, we denote by $u[x]$ the ray emanating from $x$ in direction $u$. Similarly, for a cone $C$ we denote
by $C[x]$ the translation of $C$ that places its apex at $x$.
Let $0\leq \beta\leq \pi/2$ be an angle. For a direction $u \in \sphere^1$
and for two points $p,q \in P$, we say that the edge
$e_{pq}\in \VD(P)$ is \emph{$\beta$-long around  the ray $u[q]$} 
if $p$ is the Voronoi neighbor of $q$ in all directions in the range 
$[u-\beta,u+\beta]$, i.e., for all $v\in[u-\beta,u+\beta]$, the ray 
$v[q]$ intersects $e_{pq}$.
The  {\em $\beta$-cone around  $u[q]$} is the cone whose apex is $q$
and each of its bounding rays makes an angle of $\beta$ with $u[q]$.

\begin{figure}[htbp]
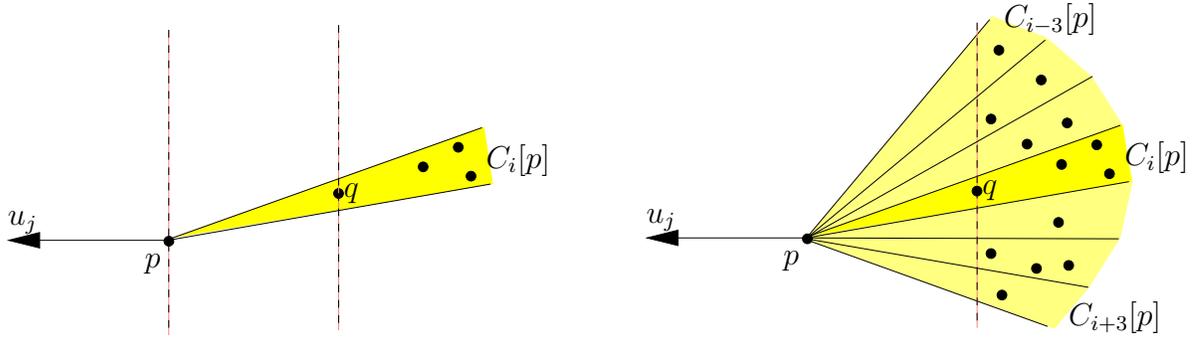

\begin{center}
\input{Jextremal1.pstex_t}\hspace{2cm}\input{StronglyJextremal1.pstex_t}
\caption{\small \sf Left: $q$ is $j$-extremal for $p$. Right: $q$ is \textit{strongly} $j$-extremal for $p$.}
\label{Fig:Jextremal}
\end{center}
\end{figure}

\paragraph{Definition ($j$-extremal points).}
(i) Let $p,q\in P$, let $i$ be the index such that $q\in C_i[p]$,
and let $u_j$ be a direction
 such that
  $\inprod{u_j}{x}\le 0$ for all $x\in C_i$.
We say that $q$ is \emph{$j$-extremal} for $p$ if 
  $q = \arg\max \{\inprod{p'}{u_j}\mid p'\in C_i[p]\cap P\setminus\{p\}\}$.
That is, $q$ is the nearest point to $p$ in this cone, in the $(-u_j)$-direction.
Clearly, a point $p$ has at most $s$ $j$-extremal points,
one for every admissible cone $C_i[p]$, for any fixed $j$. See Figure \ref{Fig:Jextremal} (left).

(ii) For $0\leq i< k$, let $C'_i$ denote the extended cone that
is the union of the seven consecutive cones $C_{i-3},\ldots,
C_{i+3}$. 
Let $p,q\in P$, let $i$ be the index such that $q\in C_i[p]$,
and let $u_j$ be a direction
such that 
 $\inprod{u_j}{x}\le 0$ for all $x\in C'_i$ (such $u_j$'s exist if $\alpha$ is smaller than some appropriate constant).
We say that the point $q\in P$ is \textit{strongly $j$-extremal} for $p$
if $q = \arg\max \{\inprod{p'}{u_j}\mid p'\in C'_i[p]\cap P\setminus\{p\}\}$. 

(iii) We say that a pair $(p,q)\in P\times P$ is {\it (strongly)} $(j,\ell)$-{\it extremal}, for some $0\leq j,\ell \leq k-1$, if $p$ is (strongly) $\ell$-extremal for
$q$ and $q$ is (strongly) $j$-extremal for $p$.

\begin{figure}[htbp]
\begin{center}
\input{necessaryCond.pstex_t} 
\caption{\small \sf Illustration of the setup in Lemma~\ref{Lemma:EmptyCone}: the edge $e_{pq}$ is $\beta$-long around $v[p]$, and the ``tip" $\triangle \sigma^+q\sigma^-$ of the cone $C[q]$ is empty.}
\label{Fig:EmptyCone}
\end{center}
\end{figure}

\begin{lemma}\label{Lemma:EmptyCone}
Let $p,q\in P$, and let $v$ be a direction such that the
edge $e_{pq}$ appears in $\VD(P)$ and is $\beta$-long around the ray 
$v[p]$. Let $C[q]$ be the $\beta$-cone around the ray from $q$ through 
$p$.  Then $\inprod{p}{v}\geq \inprod{p'}{v}$ for all
$p'\in P\cap C[q]\setminus \{q \}$.
\end{lemma}
\begin{proof}
Refer to Figure \ref{Fig:EmptyCone}. Without loss of generality, we assume that $v$ is the $(+x)$-direction
and that $q$ lies above and to the right of $p$. (In this case the slope of the bisector $\bisect_{pq}$ is negative. Note that $q$ has to
lie to the right of $p$, for otherwise $\bisect_{pq}$ would not cross $v[p]$.) 
Let $v^+$
(resp., $v^-$) be the direction that makes a counterclockwise (resp.,
clockwise) angle of $\beta$ with $v$. Let $a^+$ (resp., $a^-$) be
the intersection  of $e_{pq}$ with $v^+[p]$ (resp., with $v^-[p]$); by assumption, both points exist.
Let $h$ be the vertical line passing through $p$. Let $\sigma^+$
(resp., $\sigma^-$) be the intersection point of $h$ with the ray
emanating from $a^+$ (resp., $a^-$) in the direction opposite to
$v^-$ (resp., $v^+$); see Figure~\ref{Fig:EmptyCone}.

Note that $\angle pa^+\sigma^+=2\beta$, and  that
$\|a^+\sigma^+\|=\|pa^+\|=\|qa^+\|$, i.e., $a^+$ is the circumcenter
of $\triangle p\sigma^+q$.   Therefore $\angle \sigma^+ q p =
\frac12{\angle \sigma^+ a^+p} = \beta$. That is, $\sigma^+$ is the intersection of the upper ray of $C[q]$ with $h$.
Similarly, $\sigma^-$ is the intersection of the lower ray of $C[q]$ with $h$.
Moreover, if there exists a
 point $x\in P$ properly inside the triangle $\triangle pq\sigma^+$
then
$\|a^+x\| < \|a^+p\|$,
contradicting the fact that $a^+$ is on $e_{pq}$. So the interior of $\triangle pq\sigma^+$ (including the relative interiors of edges $pq,\sigma^+q$) is disjoint from $P$.
 Similarly, by a symmetric argument, no points of $P$ lie inside 
$\triangle pq\sigma^-$ or on the relative interiors of its edges $pq,\sigma^-q$.
Hence, the portion of $C[q]$ to the right of $p$ is openly disjoint from $P$, and therefore $p$ is a rightmost point of $P$ (extreme in the $v$ direction) inside $C[q]$.\end{proof}

\begin{corollary}\label{Corol:ExtremalPair}
Let $p,q\in P$. 
\noindent (i) If the edge $e_{pq}$ is $3\alpha$-long in $\VD(P)$
then there are $0\leq j,\ell< k$ for which $(p,q)$
is a $(j,\ell)$-extremal pair.
\noindent(ii) If the edge $e_{pq}$ is $9\alpha$-long in $\VD(P)$
then there are $0\leq j,\ell < k$ for which $(p,q)$
is a strongly $(j,\ell)$-extremal pair.
\end{corollary}
\begin{proof}
To prove part (i), choose $0\leq j, \ell < k$, such that $e_{pq}$ is
$\alpha$-long around each of  $u_{\ell}[p]$ and  $u_{j}[q]$. By Lemma \ref{Lemma:EmptyCone}, $p$ is $u_\ell$-extremal in the $\alpha$-cone $C[q]$ around the ray from $q$ through $p$. Let $i$ be the index such that $p\in C_i[q]$. Since the opening angle of $C[q]$ is $2\alpha$, it follows that $C_i[q]\subseteq C[q]$, so $p$ is $\ell$-extremal with respect to $q$, and, symmetrically, $q$ is $j$-extremal with respect to $p$. To prove part (ii) choose $0\leq
j,\ell< k$, such that $e_{pq}$ is $4\alpha$-long around each of $u_{\ell}[p]$
and $u_{j}[q]$ and apply
Lemma~\ref{Lemma:EmptyCone} as in the proof of part (i).
\end{proof}

\paragraph{The stable Delaunay graph.}
We kinetically maintain a
$(10\alpha,\alpha)$-stable Delaunay graph, whose precise definition is given below,
using a data-structure which is based on a collection of 2-dimensional orthogonal range trees similar to the ones used in \cite{KineticNeighbors}.

Fix $0\leq i< s$, and choose a ``sheared" coordinate frame in
 which the rays $u_i$ and $u_{i+1}$ form the $x$- and $y$-axes,
 respectively. That is, in this coordinate frame, $q\in C_i[p]$ if and
 only if $q$ lies in the upper-right quadrant anchored at $p$. 
 
We define a 2-dimensional range tree $\T_i$ consisting of a \textit{primary} balanced binary search tree with the
points of $P$ stored at its leaves ordered by their $x$-coordinates, and of secondary trees, introduced below.
Each internal node $v$ of the primary tree of $\T_i$ is associated with the
\emph{canonical subset} $P_v$ of all points that are stored at the
leaves of the subtree rooted at $v$. A point $p\in P_v$ is said to be
\emph{red} (resp., {\it blue) in $P_v$} if it is stored at the subtree rooted at
the left (resp., right) child of $v$ in $\T_i$. For each primary node
$v$ we maintain a secondary balanced binary search tree $\T_i^v$, whose
leaves store the points of $P_v$ ordered by their $y$-coordinates.
We refer to a node $w$ in a secondary
tree  $\T_i^v$  as a {\em secondary node $w$ of $\T_i$}.

Each node $w$ of a secondary tree $\T_i^v$ is associated with a canonical
subset $P_w\subseteq P_v$ of points stored at the leaves of the
subtree of $\T_i^v$  rooted at $w$. We also associate with $w$ the sets
$R_w\subset P_w$ and $B_w\subset P_w$ of points  residing in the
\emph{left} (resp., \emph{right}) \emph{subtree} of $w$ and are red
(resp., blue) in $P_v$.  It is easy to verify that the
 sum of the sizes of the sets $R_w$ and $B_w$
over all secondary nodes of $\T_i$
 is $O(n\log^2n)$.

For each secondary
node $w\in \T_i$ and each $0\leq j< k$ we maintain the points 
$$
\xi^R_{i,j}(w)=\arg \max_{p\in
R_w}\inprod{p}{u_j}, \quad \xi^B_{i,j}(w)=\arg \max_{p\in B_w}\inprod{p}{u_j},
$$ provided that both $R_w,B_w$ are not empty. 
See Figure \ref{Fig:JLextremalNode}. It is straightforward to show that if $(p,q)$ is a $(j,\ell)$-extremal pair, so that $q\in C_i[p]$,
then there
is a secondary node $w\in \T_i$ for which $q=\xi^B_{i,j}(w)$
 and $p=\xi^R_{i,\ell}(w)$.
 
\begin{figure}[htbp]
\begin{center}
\input{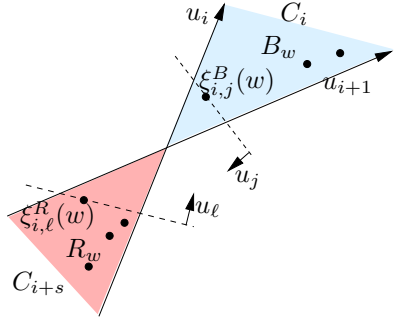}
\caption{\small \sf The points $\xi_{i,\ell}^R(w),\xi_{i,j}^B(w)$ for a secondary node $w$ of $\T_i$.}
\label{Fig:JLextremalNode}
\end{center}
\end{figure}

For each $p\in P$ we construct a set $\NN[p]$ containing all points $q\in
P$ for which $(p,q)$ is a $(j,\ell)$-extremal pair, for some pair of indices $0\leq j,\ell< k$. Specifically, for
each $0\le i < s$, and each secondary node $w\in \T_i$ such that
$p=\xi^R_{i,\ell}(w)$ for some $0 \le \ell < k$, we include in
$\NN[p]$ all the points $q$ such that $q=\xi^B_{i,j}(w)$ for some $0 \le
j < k$. Similarly,
for each $0\le i < s$, and each secondary node
$w\in \T_i$ such that $p=\xi^B_{i,\ell}(w)$ for some $0 \le \ell <
k$ we include in $\NN[p]$ all the points $q$ such that
$q=\xi^R_{i,j}(w)$ for some $0 \le j < k$.
It is easy to verify that, for each $(i,\ell)$-extremal pair $(p,q)$, for some $0\leq j,\ell<k$,
$q$ is placed in $\NN[p]$ by the preceding process. The converse, however, does not always hold, so in general $\{p\}\times\NN[p]$ is a superset of
the pairs that we want. 

For each $0\leq i< s$, each
point $p\in P$ belongs to $O(\log^2n)$ sets $R_w$ and $B_w$, so the size of
$\NN[p]$ is bounded by $O(s^2\log^2n)$. Indeed, $p$ may be coupled with up to $k=2s$ neighbors at each of the $O(s\log^2n)$ nodes containing it.

For each point $p\in P$ and $0\leq \ell < k$ we 
maintain all points in $\NN[p]$ in a kinetic and dynamic
tournament $\dirtour_\ell[p]$ whose winner $q$ minimizes
the directional distance $\varphi_\ell[p,q]$, as given in (\ref{Eq:DirectDist}). That is, 
the winner in $\dirtour_\ell[p]$ is $\Nbrs_\ell[p]$ in the Voronoi diagram of
$\{p\}\cup \NN[p]$. 

We are now ready to define the stable Delaunay graph $\G$ that we maintain.
For each pair of points $p,q\in
P$ we add the edge $(p,q)$ to $\G$ if 
the following hold.
\begin{itemize}
\item[(G1)] There is an index $0\leq \ell < k$ such that $q$ wins the 8
consecutive tournaments $\dirtour_\ell[p],\ldots,\dirtour_{\ell+7}[p]$.
\item[(G2)] The point $p$ is strongly $(\ell+3)$-extremal and strongly
$(\ell+4)$-extremal for $q$.
\end{itemize}

The $(10\alpha,\alpha)$-stability of $\G$ is implied by a combination of Theorems \ref{Thm:Completeness} and \ref{Thm:Soundness}.
\begin{theorem}\label{Thm:Completeness}
For every $10\alpha$-long edge $e_{pq}\in \VD(P)$, the graph $\G$
contains the edge $(p,q)$.
\end{theorem}
\begin{proof}
By Corollary~\ref{Corol:ExtremalPair} (i), 
there are $j$ and $\ell$ such that $(p,q)$ is a $(j,\ell)$-extremal pair.
By the preceding discussion this implies that $q$ is in $\NN[p]$.
 Now since $e_{pq}$ is  $10\alpha$-long 
  there is an $\ell'$ such that 
$N_{\ell'}[p],\ldots,N_{\ell'+7}[p]=q$ in $\VD(P)$, and therefore also
in the Voronoi diagram of $\{p\}\cup \NN[p]$.
So it follows that $q$ indeed wins the tournaments
 $\dirtour_{\ell'}[p],\ldots,\dirtour_{\ell'+7}[p]$.

By the proof of Corollary~\ref{Corol:ExtremalPair} (ii), $p$ is
strongly  $(\ell'+3)$-extremal and
strongly $(\ell'+4)$-extremal for $q$.
\end{proof}

\begin{theorem}\label{Thm:Soundness}
For every edge $(p,q)\in \G$, the edge $e_{pq}$ belongs to $\VD(P)$ and is $\alpha$-long there.
\end{theorem}
\begin{proof}

Since $(p,q)\in \G$ we know that $q$ is in $\NN[p]$ and wins the tournaments
$\dirtour_\ell[p],\dirtour_{\ell+1}[p],\ldots,\dirtour_{\ell+7}[p]$, for some $0\leq \ell <k$ and
that the point $p$ is strongly $(\ell+3)$-extremal and $(\ell+4)$-extremal for $q$. 
We prove that the rays $u_{\ell+3}[p]$ and $u_{\ell+4}[p]$ stab $e_{pq}$,
from which the theorem follows.

Assume then that
 one of the rays $u_{\ell+3}[p], u_{\ell+4}[p]$ does not stab
 $e_{pq}$; suppose it is the ray $u_{\ell+4}[q]$. (This includes the case when $e_{pq}$ is not present at
 all in $\VD(P)$.)  By definition, this means that
 $r=N_{\ell+4}[p]\neq q$. We use Lemma~\ref{Lemma:qExtremal}, given shortly below, to show that
 $q$ cannot win in at least one of the tournaments among
 $\dirtour_{\ell}[p],\ldots,\dirtour_{\ell+7}[p]$ and thereby get a
 contradiction.

\begin{figure}[htbp]
\begin{center}
\input{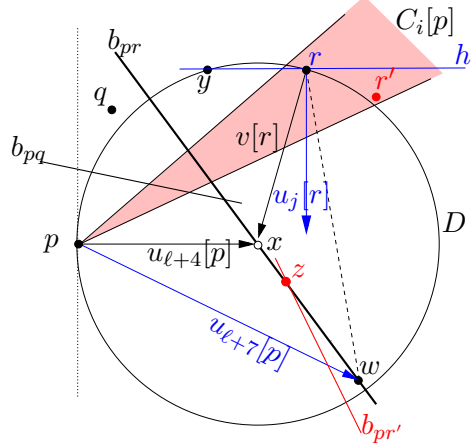}
\end{center}
\caption{\sf \small Proof of Theorem~\ref{Thm:Soundness}: the case when $r$
is to the right of the line from $p$ to $q$. The line $h$
orthogonal to $u_j$ through $r$
intersects the circle $D$ at a point $y$ outside
$C_i[p]$, which implies that  $r'$ 
is to the right of the line from $p$ to $r$.
 Assuming $r\neq r'$, the point $z=\bisect_{pr}\cap
\bisect_{pr'}$ is inside the cone bounded by $u_{\ell+4}[p]$ and
$u_{\ell+7}[p]$. Hence, $u_{\ell+7}[p]$ hits $\bisect_{pr'}$ before
$\bisect_{pr}$.} \label{Fig:Clockwise}
\end{figure}

According to Lemma~\ref{Lemma:qExtremal},
there exists a point $r$ such that $\varphi_{\ell+4}[p,r] < \varphi_{\ell+4}[p,q]$
and  
$p$ is $(\ell+4)$-extremal for $r$. 
Let  $x=u_{\ell+4}[p]\cap \bisect_{pr}$ and let
$D$ be the circle which is centered at $x$, and passes through $r$ and $p$; see Figure \ref{Fig:Clockwise}.


We consider the case where $r$ is to the right of the  line
from $p$ to $q$; 
the other case is treated symmetrically.  In this case the
intersection of $\bisect_{pr}$ and $\bisect_{pq}$ is to the left of
the directed line from $p$ to $x$. Let $0\leq i\leq k-1$ be the index for
which $r\in C_i[p]$.  If $i\le s-1$ then there is a secondary
node $w$ in the tree $\T_i$ for which $p\in R_w$ and $r\in B_w$, and
since $p$ is $(\ell+4)$-extremal for $r$, $\xi^R_{i,\ell+4}(w)$ is equal to
$p$. If $i > s$ then, symmetrically, we have a node $w\in \T_{i-s}$
such that $r\in R_w$ and $p\in B_w$ and $\xi^B_{i,\ell+4}(w)$ is equal to
$p$. We assume that $i\le s-1$ in the sequel; the other case is treated in a fully symmetric manner.

Let $v[r]$ be the ray from $r$ through $x$, for an appropriate direction $v\in \S^1$, and let $u_j$ be the direction which lies counterclockwise to $v$ and forms with it an angle of at least $\alpha$ and at most $2\alpha$.
Put
$r'=\xi^B_{i,j}(w)$, implying that $r'\in C_i[p]$ and
$\inprod{r'}{u_j}\geq \inprod{r}{u_j}$. 
In particular, $r'$ belongs
to $\NN[p]$. 
If $r'$ is inside $D$ (and in particular if $r'=r$) then $q$ cannot win 
 the tournament $\dirtour_{\ell+4}[p]$ which is the contradiction we are after.
So we may assume that $r'$ is outside $D$.


Let $h$ be the line through $r$ orthogonal to $u_j$.  
Clearly, $h$ intersects $D$ at two points, $r$ and another point
$y$ (lying counterclockwise to $r$ along $\partial D$, by the choice of $u_j$). Since $\angle rpy=\frac{1}{2}\angle rxy$, and $\angle rxy$ equal to twice the angle between $v$ and $u_j$, $\angle rpy$ is at least $\alpha$,
 so $y$ is outside $C_i[p]$.  By assumption, $r'$ lies in the halfplane bounded by $h$ and containing $p$. Since we assume that $r'$
is not in $D$ it must be to the right of the  line
from $x$ to $r$.  It follows that $b_{pr'}$ intersects $b_{pr}$ at
some point $z$ to the right of the  line from $p$ to $x$; see
Figure~\ref{Fig:Clockwise}.

We claim that $z$ is inside the cone with apex $p$ bounded by the rays
$u_{\ell+4}[p]$ and $u_{\ell+7}[p]$. 
Indeed, suppose to the contrary that the claim is false.
It follows that in the diagram $\VD(\{r,r',p\})$ the edge
$e_{pr}$ is $\alpha$-long around $u_j[r]$. Indeed, denote the intersection point of $u_{\ell+7}[p]$ and $\bisect_{pr}$ as $w$ (see Figure \ref{Fig:Clockwise}). Then $\angle xrw=\angle xpw= 3\alpha$. Since the angle between $v[r]$ and $u_j[r]$ is between $\alpha$ and $2\alpha$, the claim follows. Now, according to
Lemma~\ref{Lemma:EmptyCone}, $\inprod{r}{u_j}\leq\inprod{r'}{u_j}$,
which contradicts the choice of $r'$. It follows that $z$ is in
the cone bounded by $u_{\ell+4}[p]$ and $u_{\ell+7}[p]$ and thus
$u_{\ell+7}[p]$ hits $b_{pr'}$ before $b_{pr}$, and therefore also before $\bisect_{pq}$. Hence, $q$ cannot win $\dirtour_{\ell+7}[p]$, and we get the final contradiction which completes the proof of the theorem.
\end{proof}

\begin{figure}[htbp]
\begin{center}
\input{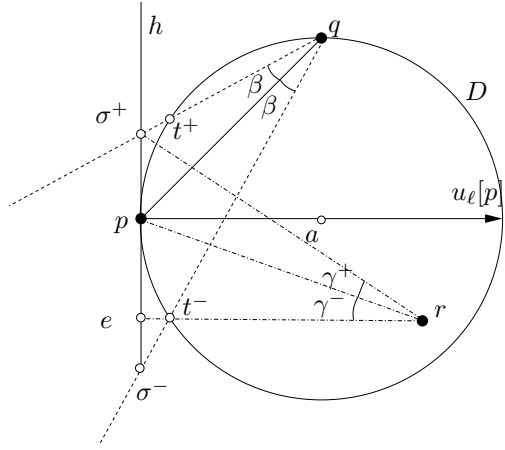}
\caption{\sf \small The proof of Lemma \ref{Lemma:qExtremal}: The point $p$ is strongly $\ell$-extremal for $q$ and
$\ell$-extremal for $r$.}\label{Fig:Extremal1}
\end{center}
\end{figure}

\noindent {\bf Remark:} We have not made any serious attempt to reduce the constants $c$ appearing in the definitions of various $(c\alpha,\alpha)$-$\SDG$s that we maintain. We suspect, though, that they can be significantly reduced.

To complete the proof of Theorem \ref{Thm:Soundness}, we provide the missing lemma.

\begin{lemma}\label{Lemma:qExtremal}
Let $p,q\in P$ be a pair of points and $0\leq \ell \leq k-1$ an index,
such that the point $p$ is strongly $\ell$-extremal for $q$ but
$N_\ell[p]\neq q$.
Then there exists a point $r$ such that $\varphi_\ell[p,r] < \varphi_\ell[p,q]$ and
 $p$ is $\ell$-extremal
for  $r$.
\end{lemma}

\begin{proof}

Let $0\leq i\leq k-1$ be the index for which $q\in C_i[p]$ and let
$h$ be the 
line through $p$, orthogonal to $u_\ell$.
Assume without loss of generality that $h$ is vertical and
the ray $u_\ell[p]$ extends to the right of $h$.

\begin{figure}[htb]
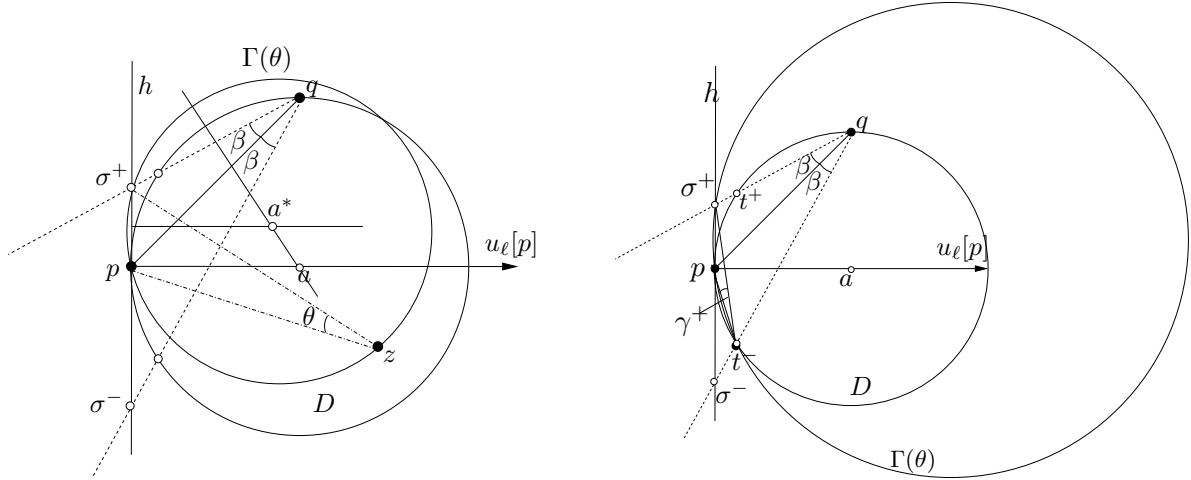

\begin{center}
\input{fix3.pstex_t}\hspace{1cm}\input{fix4.pstex_t}
\caption{\sf \small Left: The circular arc $\Gamma(\theta)$ is the locus of all points
which are to the right of $p\sigma^+$ and see it at angle $\theta$. Right: To minimize $\theta$ we increase the radius of
$\Gamma(\theta)$ until one of its intersection points with $D$
coincides with $t^-$.}\label{Fig:Extremal2}
\end{center}
\end{figure}

Let $a$ be the point at which $u_\ell[p]$ intersects the bisector
$\bisect_{pq}$, and let $D$ be the disk centered at $a$
whose boundary contains both $p$ and $q$. Since
$N_\ell[p]\neq q$, the interior of $D$ must contain some other point $r\in
P$; see Figure~\ref{Fig:Extremal1}. 

Let $C[q]$ be the cone 
emanating from $q$ such that each of its bounding rays makes an angle of
$\beta=3\alpha$ with 
the ray from $q$ through $p$;
 in particular $C[q]$ contains $p$. 
Let $\sigma^+$ (resp., $\sigma^-$) denote the upper (resp., lower) endpoint of the
intersection of  $C[q]$ and  $h$. Since $p$ is strongly $\ell$-extremal for
$q$, the interior of the triangle $\triangle \sigma^+q\sigma^-$ does not contain any
points of $P$. Hence, $r$ must be outside 
the triangle $\triangle \sigma^+q\sigma^-$. So either $r$ is above
 $q\sigma^+$ (and inside $D$) or below $q\sigma^-$ (and inside $D$).

Assume, without loss of generality, that $r$ is below $q\sigma^-$, as shown in
Figure~\ref{Fig:Extremal1}. (The case where $r$ is above
 $q\sigma^+$ is fully symmetric.) 
 Let $t^+$ and $t^-$ denote the intersection points
$q\sigma^+\cap\bd D$ and $q\sigma^-\cap\bd D$, respectively. Let $e$
be the point at which the ray 
from $r$ through $t^{-}$
intersects $h$. Then
the intersection of the triangle $\triangle \sigma^+re$ and $\triangle
\sigma^+q\sigma^-$ is empty.
Among all the points of $P$ in $D$ we choose $r$ so that its $x$-coordinate is
the smallest. For this choice of $r$ we also
have that $\triangle \sigma^+re \setminus \triangle
\sigma^+q\sigma^-$ is empty (since it is contained in $D$ and lies to the left of $r$).
In other words, $\triangle \sigma^+re$ is empty.

Let $\gamma^+$ (resp., $\gamma^-$) denote the angle $\angle pr\sigma^+$
(resp., $\angle pr t^{-}$). It remains to show that
$\gamma^+\geq \frac{1}{3}\beta$ and $\gamma^-\geq \frac{1}{3}\beta$.
This will imply that the cone $C_{i'}[r]$ that contains $p$ is fully contained in the cone bounded by the rays from $r$ through $\sigma^+$ and $t^-$, so $p$ is extreme in the $u_\ell$-direction within $C_{i'}[r]$, which is what the lemma asserts. Since $r$ is inside $D$, it is clear that
$\gamma^-\ge \angle pqt^{-}=\beta$. 
The angle
$\gamma^+$ however may be smaller than $\beta$, but, as we next show,
$\tan\gamma^+ \ge \frac13\tan \beta$.
Indeed, fix an angle $\theta$ and
let $\Gamma(\theta)$ denote the circular arc which is the locus of
all points $z$ that are to the right of $h$ and 
the angle $\angle pz\sigma^+$ is 
$\theta$. The endpoints of $\Gamma(\theta)$ are $p$ and
$\sigma^+$, and its center $a^*$ is on the (horizontal) bisector of
$p\sigma^+$; see Figure~\ref{Fig:Extremal2} (left). 

Notice that $\Gamma(\theta)$ intersects $\bd D$ at two points, one of which is $p$, which
are symmetric with respect to the line 
through $a$ and $a^*$.  As $\theta$ decreases
$a^*$ moves to the right, and the intersection of $\Gamma(\theta)$
with $\bd D$ rotates clockwise around $\bd D$.
Consider the smallest $\theta$ such that $\Gamma(\theta)$ intersects $D$ on or below  $qt^{-}$. It follows that this intersection is at $t^-$.
 See Figure~\ref{Fig:Extremal2} (right).

This shows that
for  fixed  $p$ and $q$, the position of $r$
in $D$ below the line $qt^{-}$ which minimizes $\gamma^+$ is at
$t^{-}$. 
To complete the analysis, we look for the position of $q$ that minimizes $\gamma^+$ when
$r$ is at $t^{-}$. 
Note that, as $q$ moves along $\partial D$, the points $t^+$ and $t^-$ do not change.
As shown in Figure~\ref{Fig:Extremal3} (left),
 $\gamma^+$ decreases when 
 $q$ tends counterclockwise to $t^+$.
When $q$ is at $t^{+}$,
 $q\sigma^+$ is tangent to $D$.
A simple calculation, 
 illustrated in Figure~\ref{Fig:Extremal3} (right), shows that
 $\tan\gamma^+ = \frac13 \tan \beta$. By the
 inequality $\tan (3x) > 3\tan x$, for $x$ sufficiently small, it follows that $\gamma^+ >
 \frac{1}{3}\beta$, implying, as noted above, that the point $p$ is $\ell$-extremal 
 for $r$. This completes the proof of the lemma.
\end{proof}

\begin{figure}[htbp]
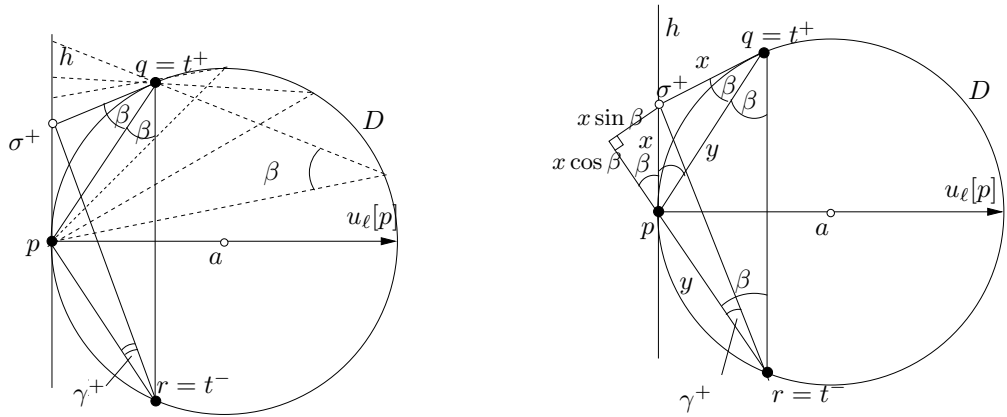

\begin{center}
\input{fix5.pstex_t}\hspace{2cm}\input{fix6.pstex_t} 
\caption{\sf \small Left: $\gamma^+$ is minimized as $q$ tends counterclockwise to $t^+$.
Right: Proving that $\tan \gamma^+=\frac{1}{3}\tan \beta$ when $q=t^+$ and
$r=t^-$. The triangles $\triangle q\sigma^+p$ and $\triangle pqr$ are isosceles and similar, and $y=2x\cos\beta$. Thus $\tan
\gamma^+=\frac{x\sin\beta}{x\cos\beta+y}=\frac{1}{3}\tan\beta$. }\label{Fig:Extremal3}
\end{center}
\end{figure}

In Section \ref{Subsec:ReducedMaintenNaive} we describe a naive algorithm for kinetic maintenance of $\G$, which encounters a total of $O^*(k^4n^2)$ events in the tournaments $\dirtour_\ell[p]$. In Section \ref{Subsec:ReducedMaintenImprove} we consider a slightly more economical definition of the tournaments $\dirtour_\ell[p]$, yielding a solution which processes only $O^*(k^2n^2)$ events in $O^*(k^2n^2)$ overall time.
\subsection{Naive maintenance of $\G$}\label{Subsec:ReducedMaintenNaive}
As the points of $P$ move, we need to update the $\SDG$ $\G$, which, as we recall, contains those edges $(p,q)$ such that $q$ wins $8$ consecutive tournaments $\dirtour_{\ell}[p],\ldots,\dirtour_{\ell+7}[p]$ of $p$, and $p$ is strongly $(\ell+3)$-extremal and $(\ell+4)$-extremal for $q$. We thus need to detect and process instances at which one of these conditions changes. There are several events at which such a change can occur:\\
\indent(a) A change in the sets of neighbors $\NN[p]$, for $p\neq P$.\\
\indent(b) A change in the status of being strongly $\ell$-extremal for some pair $(p,q)$.\\
\indent(c) A change in the winner of some tournament $\dirtour_\ell[p]$ (at which two existing members of $\NN[p]$ attain the same minimum distance in the direction $u_\ell$).

Note that each of the events (a)--(b) can arise only during a swap of two points in one of the $s$ directions $u_0,\ldots ,u_{s-1}$ or in one of the directions orthogonal to these vectors. 

For each $0\leq i\leq s-1$ we maintain two lists. The first list, $L_i$, stores the
points of $P$ ordered by their projections on a line in the $u_i$-direction, and the second list, $K_i$, stores the points ordered by their projections on a line orthogonal to the $u_i$-direction. We note that, as long as the order in each of the $2s$ lists $K_i,L_i$ remains unchanged, the discrete structure of the range trees $\T_i$, and the auxiliary items $\xi_{i,\ell}^R(w),\xi_{i,j}^B(w)$, does not change either. More precisely, the structure of $\T_i$ changes only when two consecutive elements in $K_i$
or in $K_{i+1}$ swap their order in the respective list; whereas the auxiliary items $\xi_{i,j}^R(w),\xi_{i,j}^B(w)$, stored at secondary nodes of $\T_i$, may also change when two consecutive points swap their order in the list $L_j$. 
There are $O(sn^2)=O(n^2)$ discrete events where consecutive points in $K_i$ or $L_i$
swap. We call these events {\em $K_i$-swaps} and \textit{$L_i$-swaps}, respectively. Each such event happens
when the line trough a pair of points becomes orthogonal or parallel to $u_i$. We
can maintain each list in linear space for a total of $O(sn)$ space
for all lists. Processing a swap takes $O(\log n)$ time to replace a
constant number of elements in the event queue (and more time to update the various structures, as discussed next).

\smallskip

\noindent{\bf The range trees $\T_i$.}
As just noted, the structure of $\T_i$ changes either at a $K_i$-swap or at a
$K_{i+1}$-swap. As described in \cite[Section 4]{KineticNeighbors}, we
can update $\T_i$ when such a swap occurs, including the various auxiliary data that it stores, in $O(s\log^2 n)$ time. 
(The factor $s$ is due to the fact that we maintain
$O(s)$ extreme points $\xi_{i,\ell}^B(w)$ and $\xi_{i,j}^R(w)$ in each
secondary node $w$ of $\T_i$, whereas in \cite{KineticNeighbors} only
two points are maintained.)

In a similar manner, an $L_j$-swap of two points $p,q$ may affect one of the items $\xi_{i,j}^B(w)$
and $\xi_{i,j}^R(w)$ stored at any secondary node $w$ of any $\T_i$, for $0\leq i\leq s-1$, such that both $p,q$ belong to $R_w$ or to $B_w$. Each $\T_i$ has only $O(\log^2n)$ such nodes, and the data structure of \cite{KineticNeighbors} allows us to
update $\T_i$, when an $L_j$-swap occurs in $O(\log^2 n)$ time. Summing up over all $0\leq i\leq s-1$, we get
that the total update time of the range trees after an $L_j$-swap is
$O(s\log^2 n)$.  As follows from the analysis in~\cite[Section 4]{KineticNeighbors}, the
trees $\T_i$, for $0\leq i\leq s-1$, require a total of $O(s^2n\log
n)$ storage (because of the $O(s)$ items
$\xi_{i,\ell}^B(w),\xi_{i,j}^R(w)$ stored at each secondary node of each of the $s$ trees).

\smallskip

\noindent{\bf The tournaments $\dirtour_\ell[p]$.}
The kinetic tournament $\dirtour_\ell[p]$, for $p\in P$ and $0\leq
\ell\leq k-1$ contains the points in the set $\NN[p]$.  Since $\NN[p]$
varies both kinetically and dynamically and therefore the tournaments
$\dirtour_\ell[p]$ need to be maintained as kinetic and dynamic tournaments, in the manner reviewed in Section \ref{sec:Prelim}. 

For $0\leq i\leq s-1$, we define $\Pi_{i}$ to be
the set of pairs of points $(p,q)$, such that there exists a
secondary node $w$ in $\T_i$, and indices $0\leq j,\ell\leq k-1$, for which
$p=\xi_{i,\ell}^R(w)$ and $q=\xi_{i,j}^B(w)$.
For a fixed $i$, a point $p$ belongs to $O(s\log^2n)$ pairs $(p,q)$ in
$\Pi_{i}$, for a total of $O(s^2\log^2n)$ pairs over all sets $\Pi_{i}$.  It follows that the total size of all the
sets $\Pi_i$ is $O(s^2n\log^2n)$.  Any secondary node of any tree
$\T_i$, for $0\leq i\leq s-1$, contributes at most $O(s^2)$ pairs to
the respective set $\Pi_{i}$.

The set $\NN[p]$ consists of all the points $q$ such that there exists a set
$\Pi_{i}$ that contains the pair $(p,q)$ or the pair $(q,p)$.  So the
total size of the sets $\NN[p]$, over all points $p$, is $O(s^2 n\log^2
n)$. A set $\NN[p]$ changes only when one of the sets $\Pi_{i}$ changes,
which can happen only as the result of a swap.

Specifically, when $\xi_{i,\ell}^R(w)$ changes for some $0\leq i\le
s-1$ and $0\le \ell \leq k-1$, from a point $p$ to a point $p'$, we
make the following updates.  
(i) If $p\not= \xi_{i,\ell'}^R(w)$ for all
$\ell' \not= \ell$ then for every $0\le j \le k-1$ we delete the pair
$(p,\xi_{i,j}^B(w))$ from $\Pi_{i}$. (ii) We add the
pair $(p',\xi_{i,j}^B(w))$ to $\Pi_{i}$.  We make analogous updates
when one of the values $\xi_{i,j}^B(w)$ changes.  When a node $w$ is created, deleted,
or involved in a rotation, we update the pairs ($\xi_{i,\ell}^B(w)$,
$\xi_{i,j}^R(w)$) in $\Pi_i$ for every $\ell$ and $j$. In such a case
we say that node $w$ {\em is changed}.

A change of $\xi_{i,\ell}^R(w)$ or $\xi_{i,j}^B(w)$ in an existing
node $w$ generates $O(s)$ changes in $\Pi_{i}$ and thereby $O(s)$
changes to the sets $\NN[p]$.  Thus, it may generate $O(s^2)$ updates to the
tournaments $\dirtour_\ell[p]$. A change of a secondary node may
generate $O(s^2)$ changes to the sets $\NN[p]$ and thereby $O(s^3)$ updates
to the tournaments $\dirtour_\ell[p]$.

A point $\xi_{i,\ell}^R(w)$ or $\xi_{i,\ell}^B(w)$ changes during either
a $K_i$, $K_{i+1}$, or $L_\ell$-swap.  Each $L_\ell$-swap, for any $\ell$, causes $O(s\log^2n)$
points $\xi_{i,\ell}^R(w)$ or $\xi_{i,\ell}^B(w)$ to change (over the entire collection of trees), and
therefore each swap causes $O(s^3 \log^2 n)$ updates to the tournamnets
$\dirtour_\ell[p]$.  The number of nodes which
change in $\T_i$ by a $K_i$ or $K_{i+1}$-swap is $O(\log^2n)$. Each
such change causes $O(s^3)$ updates to the tournaments $\dirtour_\ell[p]$.
Therefore the total number of updates to tournaments due to changes of
nodes is also $O(s^3\log^2n)$ per swap.

The number of swaps is $O(sn^2)$, so overall we get $O(s^4n^2\log^2n)$
updates to the tournaments.  The size of each individual tournament is
$O(s^2 \log^2 n)$.
By Theorem \ref{thm:kinetic-tour} these updates generate
$$
O(s^4n^2\log^2n \cdot\beta_{r+2}(s^2\log^2 n)\log(s^2\log^2n))
=O(s^4n^2\beta_{r+2}(s\log n)\log^2n\log(s\log n))
$$
tournament
events, which are processed in
$$
O(s^4n^2\log^2 n\cdot\beta_{r+2}(s^2\log^2 n)\log^2(s^2\log^2n))
=O(s^4n^2\cdot \beta_{r+2}(s\log n)\log^2n\log^2(s\log n))
$$ 
time.
Processing each individual tournament event takes $O(\log^2\log
n+\log^2 s)$ time.
 
Since the size of each tournament is $O(s^2\log^2n)$ and there are $O(ns)$ tournaments, the
total size of all tournaments is $O( s^3 n\log^2 n)$.

\smallskip

\noindent{\bf Testing whether $p$ is strongly $\ell$-extremal for the
  winner of $\dirtour_\ell[p]$.}
For each $0\leq i\leq s-1$, and for each pair $(p,q)\in \Pi_{i}$ we maintain
those indices $0\le \ell\le k-1$ (if there are any) for which
$p$ is  strongly $\ell$-extremal for $q$.
Recall that each point $p$ belongs to $O(s^2\log^2 n)$ pairs in the sets
$\Pi_{i}$.

We use the trees $\T_j$ for $i-3 \le j\le i+3$ to find, for a query 
$q$, the  point $\arg\max_{q'\in P\cap
  C'_{i}[q]}\inprod{q'}{u_\ell}$, for each $0\le \ell\le k-1$. The query time is $O(s\log^2n)$ 
Using this information we easily determine,
 for a pair $(p,q)$,
for which values of $\ell$
$p$ is  strongly $\ell$-extremal for $q$.


As explained above, every swap changes
$O(s^2\log^2n)$ pairs of the sets $\Pi_{i}$.
When a new pair is added to a set  $\Pi_{i}$ we query the trees
$\T_j$, $i-3 \le j\le i+3$, to find for which values of $\ell$, $p$ is
strongly $\ell$-extremal for $q$ (and vice versa).
This takes a total of $O(s^3\log^4n)$ time for each
swap.

Furthermore, a point $p$ can cease (or start) being  strongly
$\ell$-extremal for $q$ only during a swap which involves either $p$ or $q$. 
So when we process a swap 
between $p$ and some other point we
recompute, for all pairs $(p,x)$ and $(x,p)$ in the current 
sets $\Pi_{i}$ and for every
$0\le \ell \le k-1$,
 whether $p$ is  strongly $\ell$-extremal for $x$, and whether $x$ remains  strongly $\ell$-extremal for $p$.
This adds an overhead of
$O(s^3\log^4n)$ time at each
swap. 

The following theorem summarizes the results obtained so far in this section.

\begin{theorem}
The $\SDG$ $\G$ can be maintained using a data structure which
requires 
$O\left(\left(n/\alpha^3\right) \log^2 n\right)$ space
 and encounters two types of
events: swaps and tournament events.\\ 
There are $O(n^2/\alpha)$ swaps, 
each processed in $O\left(\log^4n/\alpha^3\right)$ time. 
There are
$$
 O\left(\left(n^2/\alpha^4\right)\log^2n \beta_{r+2}(\log n/\alpha)\log(\log n/\alpha)\right)
$$  tournament 
events which are processed in overall
$$
O\left(\left(n^2/\alpha^4\right)\log^2 n\beta_{r+2}(\log n/\alpha)\log^2(\log n/\alpha)\right)
$$ 
time.
Processing each individual tournament event takes  $O(\log^2\log n+\log^2 (1/\alpha))$ time.
\end{theorem}

\subsection{An even faster data structure}\label{Subsec:ReducedMaintenImprove}

We next reduce the overall time and space required to maintain $\G$
roughly by factors of $s^2$ and $s$, respectively (bringing the dependence on $s$ of both bounds down to roughly $s^2$).  We achieve that by
restricting each tournament $\dirtour_\ell[p]$ to contain a carefully chosen
subset $\NN_\ell[p]\subseteq \NN[p]$ of size $O(s\log^2n)$ (recall that the size of the entire set $\NN[p]$ is $O(s^2\log^2n)$).
  The definition of $\NN_\ell[p]$ is based on
the following lemma. Its simple proof is given in Figure \ref{Fig:Compatible}.

\begin{lemma}\label{Lemma:Compatible}
Let $p,q\in P$ and let $i$ be the index for which 
$q \in C_i[p]$. Let $0\leq \ell\leq k-1$ be an index, and $v\in \mathbb{S}^1$ a direction such that
the rays $u_\ell[p]$ and $v[q]$ intersect
$\bisect_{pq}$ at the same point. 
Then $v$ lies in one of the two consecutive cones $C_{\zeta(i,\ell)},C_{\zeta(i,\ell)+1}$, where $\zeta(i,\ell)=2i+s-\ell$.
\end{lemma}

\begin{figure}[htbp]
\begin{center}
\input{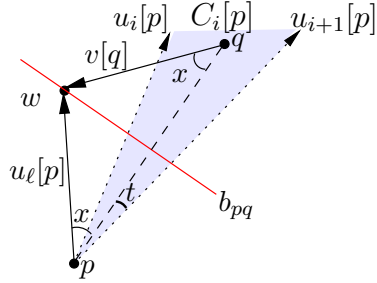}
\caption{\sf \small Proof of Lemma \ref{Lemma:Compatible}: We assume that $q\in C_i[p]$, and that the rays $u_\ell[p]$ and $v[q]$ hit $\bisect_{pq}$ at the \textit{same point} $w$. Then the angle $x=\angle wpq=(i+1-\ell)\alpha-t$, for some $0\leq t\leq \alpha$. The orientation of $\overline{qp}$ is $(i+1)\alpha-t+\pi=(i+s+1)\alpha-t$. Hence, the orientation of $v$ is $(i+s+1)\alpha-t+x=(2i+s-\ell+2)\alpha-2t$. Thus, the direction $v$ lies in the union of the two consecutive cones $C_{\zeta(i,\ell)},C_{\zeta(i,\ell)+1}$, for $\zeta(i,\ell)=2i+s-\ell$.}\label{Fig:Compatible}
\end{center}
\end{figure}

It follows that in Corollary~\ref{Corol:ExtremalPair},
we can require that the indices $0\leq j,\ell\leq k-1$, for which
$(p,q)$ is a (strongly) $(j,\ell)$-extremal pair, satisfy 
$\zeta(i,\ell) \le j \le \zeta(i,\ell)+2$.  
Indeed, we may require that the vectors $u_j[q],u_{\ell}[p]$ hit $\bisect_{pq}$ at the respective points $x$ and $y$ for which the angle $\angle xpy=\angle xqy$ is at most $\alpha$, which, in turn, happens only if $u_j$ bounds one of the cones $C_{\zeta(i,\ell)},C_{\zeta(i,\ell)+1}$.

For all $0\le i \le s-1$ and $0\le \ell\leq k-1$ we define a set $\Pi_{i,\ell}$ which
consists of all pairs $(p,q)$ of points of $P$ such that there exists a
secondary node $w$ in $\T_i$, and indices $\ell$ and $\zeta(i,\ell)
\le j \le \zeta(i,\ell)+2$, such that $p=\xi_{i,\ell}^B(w)$ and
$q=\xi_{i,j}^R(w)$ or  $p=\xi_{i,\ell}^R(w)$
and $q=\xi_{i,j}^B(w)$.  We define the set $\NN_\ell[p]$ to consist of
all points $q$ such that  $(p,q) \in \Pi_{i,\ell}$.
For a point $p$ the set of points that participate in the 
{\em reduced} tournament $\dirtour_\ell[p]$ is 
$\bigcup_{\ell'=\ell-3}^{\ell+3} \NN_{\ell'}[p]$.
(Note that this rule distributes a point $q\in \NN_{\ell}[p]$ to only seven nearby tournaments. Nevertheless, when the edge $pq$ is sufficiently long, $q$ will belong to several consecutive neighborhoods $\NN_{\ell}[p]$, and therefore will appear in more tournaments, in particular in at least eight consecutive tournaments at which it should win, according to the definition of our $\SDG$.)

We claim that, with this redefinition of the tournaments
$\dirtour_\ell[p]$, Theorems \ref{Thm:Completeness} and \ref{Thm:Soundness} still hold.
To verify that  Theorem \ref{Thm:Completeness} holds one has to follow its (short) proof and
notice that, by Lemma 
\ref{Lemma:Compatible}, the point $q$  belongs to the eight reduced tournaments 
which it is supposed to win.

We next indicate the changes required in the proof of Theorem
\ref{Thm:Soundness}. We use the same notation as in the original
proof of Theorem \ref{Thm:Soundness}, and recall that it assumed by contradiction that, say,
$N_{\ell+4}[p]\neq q$ even though $q$ wins the tournaments
$\dirtour_\ell[p],\dirtour_{\ell+1}[p],\ldots,\dirtour_{\ell+7}[p]$,
and the point $p$ is strongly $(\ell+3)$- and $(\ell+4)$-extremal for
$q$. 
We use
Lemma~\ref{Lemma:qExtremal} to establish the existence of some point
$r\in P$ such that $\varphi_{\ell+4}[p,r]<\varphi_{\ell+4}[p,q]$ and $p$ is
$(\ell+4)$-extremal for $r$.  Let $i$ be the index for which $r\in C_i[p]$, and let $w$ be the secondary node in $\T_i$ for which $r\in B_w$
and $p\in R_w$.  Note that $p=\xi_{i,\ell+4}^R(w)$.  We next choose an
 index $j$ such that the point $r'=\xi_{i,j}^B(w)$ either satisfies 
that
$\varphi_{\ell+7}[p,r']<\varphi_{\ell+7}[p,q]$ if $r$ is to the
right of the line from $p$ to $q$,  or that
$\varphi_{\ell+1}[p,r']<\varphi_{\ell+1}[p,q]$ if $r$ is to the
left of the line from $p$ to $q$.  To re-establish Theorem \ref{Thm:Soundness} it suffices to show that
$r'$ participates in the reduced tournament $\dirtour_{\ell+7}[p]$ (resp., $\dirtour_{\ell+1}[p]$) if $r$ is to the
right (resp., left) of the line from $p$ to $q$.

It follows from the way we defined $j$ in the original proof and from Lemma \ref{Lemma:Compatible} 
that $\zeta(i,\ell+4)-2\le j \le \zeta(i,\ell+4)-1$ 
(if $r$ is to the
right of the line from $p$ to $q$)
or $\zeta(i,\ell+4)+1\le j \le \zeta(i,\ell+4)+2$ (if $r$ is to the
left of the line from $p$ to $q$). 
So  $r' \in \NN_{\ell+4}[p]$ and therefore $r'$ does participate in the reduced tournament $\dirtour_{\ell+1}[p]$ or $\dirtour_{\ell+7}[p]$.
Indeed, the direction $v$ used in that proof lies in one of the cones $C_{\zeta(i,\ell+4)}, C_{\zeta(i,\ell+4)+1}$. The direction $u_j$ then forms an angle between $\alpha$ and $2\alpha$ with $v$, which lies counterclockwise from $v$ if $r$ lies to the right of the line from $p$ to $q$, or clockwise from $v$ in the other case. This is easily seen to imply the two corresponding constraints on $j$; see Figure \ref{Fig:Clockwise}.


We change our algorithm accordingly to maintain only the reduced tournaments. 

Now every secondary node $w$ of any range tree $\T_i$ contributes only seven pairs to each set
$\Pi_{i,\ell}$, for $0\leq \ell\leq k-1$, so the size of each such set is $O(n\log n)$. 
Since there
are $O(s^2)$ sets $\Pi_{i,\ell}$, their total size is $O(s^2 n\log
n)$.  Each pair in each $\Pi_{i,\ell}$ contributes an item to a constant
number of tournaments, so the total size of the tournaments is 
$O(s^2 n \log n)$.  Each individual tournament $\dirtour_{\ell}[p]$ is now
of size $O(s\log^2n)$, because $p$ belongs to $O(\log^2 n)$ pairs in each
set $\Pi_{i,\ell'}$ for $0\le i\le s-1$, $0\leq \ell'\leq k-1$, and $\dirtour_\ell[p]$ inherits only those points $q$ that come from pairs $(p,q)\in \Pi_{i,\ell'}$, for $0\leq i\leq s-1$ and $\ell-3\leq \ell'\leq \ell+3$.

When $\xi_{i,\ell}^B[w]$
 changes from $p$ to $p'$ for some $0\le i\le s-1$ and $0\le \ell \le
k-1$, at most a constant number of pairs
$(p,\xi_{i,j}^R(w))$ for $\zeta(i,\ell)\leq j\leq \zeta(i,\ell)+2$
are deleted from $\Pi_{i,\ell}$, and a constant number of pairs
$(p',\xi_{i,j}^R(w))$ for $\zeta(i,\ell)\leq j\leq \zeta(i,\ell)+2$
are added to $\Pi_{i,\ell}$.
Similar changes take place in $\Pi_{i,j}$ for those three indices $j$ satisfying
$\zeta(i,j)\leq \ell \leq \zeta(i,j)+2$. 
When $\xi_{i,j}^R[w]$ 
changes from $q$ to $q'$ for
some $0\leq i \le s-1$ and $0\le j \leq k-1$,
at most a constant number of pairs $(\xi_{i,\ell}^B(w),q)$ are deleted
from $\Pi_{i,j}$ for the indices $\ell$ satisfying $\zeta(i,j)\leq \ell \leq
\zeta(i,j)+2$, and a constant number of pairs
$(\xi_{i,\ell}^B(w),q')$ are added for the same values of
$\ell$.  Similarly,
at most a constant number of pairs $(\xi_{i,\ell}^B(w),q)$ are deleted
from $\Pi_{i,\ell}$ for the indices $\ell$ satisfying $\zeta(i,\ell)\leq j\leq
\zeta(i,\ell)+2$, and a constant number of pairs
$(\xi_{i,\ell}^B(w),q')$ are added for the same values of
$\ell$.

  A change
of a secondary node $w$ in the tree $\T_i$ causes
 $O(s)$ pairs in the sets $\Pi_{i,\ell}$
to change.

Any $K_i$-swap changes 
$O(\log^2 n)$ nodes in 
  $\T_i$ and thereby causes $O(s\log^2 n)$ pairs in the sets $\Pi_{i,\ell}$
to change. 
Any $L_j$-swap changes $O(s\log^2 n)$ extremal points $\xi_{i,j}^R[w]$,
$\xi_{i,j}^B[w]$ at secondary nodes $w$ of the trees $\T_i$, and thereby  causes 
 $O(s\log^2 n)$ pairs in the sets $\Pi_{i,\ell}$
to change. Since each pair in $\Pi_{i,\ell}$ contributes an item to a constant
number of tournaments it follows that 
$O(s\log^2n)$ points are inserted to and deleted from
the tournaments $\dirtour_\ell[p]$ at each swap.

According to Theorem \ref{thm:kinetic-tour}
the size of each tournament is $O(s\log^2 n)$ -- the number
of elements that it contains. So the total size of all tournaments
is $O(s^2n\log n)$. In total we get that there are 
$O(s^2n^2 \log^2 n)$ updates to tournaments during swaps.
These updates generate 
$$
O(s^2n^2\log^2n\beta_{r+2}(s\log n)\log(s\log n))$$ 
tournament events
that are processed in overall
$$
O(s^2n^2\log^2n\beta_{r+2}(s\log n)\log^2(s\log n))
$$ 
time.
Each  individual tournament event 
is processed in $O(\log^2\log n + \log^2 s)$ time and each swap
can be processed in  $O(s\log^2n\log^2(s\log n))$
time.

\smallskip

In addition, for each pair $(p,q)\in \Pi_{i,\ell}$
we record whether $p$ is strongly
$\ell$-extremal for $q$.
We maintain this information using the 
trees $\T_j$, for $i-3\leq j\leq i+3$, as described above, which allow
for any $p,q\in P$ and $0\leq \ell\leq k-1$ to test, in $O(\log^2n)$
time, if $p$ is strongly $\ell$-extremal for $q$.  At each swap event
we spend $O(s\log^4n)$ extra time to compute for
$O(s\log^2n)$ pairs $(p,q)$ which are added to the sets $\Pi_{i,\ell}$
whether $p$ is strongly
$\ell$-extremal for $q$.

Consider a pair $(p,q) \in \Pi_{i,\ell}$. The point $p$ may stop being
strongly
$\ell$-extremal for $q$ only during a swap which involves $p$ 
or $q$. So, as before, at each swap we find the $O(s\log^2n)$ pairs
containing one of the points involved in the swap, and recompute, in $O(s\log^4n)$
total time,
for each such pair $(p,q)$, whether
the strong extremal relation holds. 
 We thus obtain the following summary result.

\begin{theorem}\label{Thm:ReducedS}
Let $P$ be a set of $n$ moving points in $\reals^2$ under algebraic
motion of bounded degree, 
and let $\alpha > 0$ be a sufficiently small parameter. A $(10\alpha,\alpha)$-SDG of $P$
can be maintained using a data structure that requires 
$O((n/\alpha^2) \log n)$ space and encounters two types of
events: swap events 
and tournament events.  There are $O(n^2/\alpha)$ swap events, 
each processed in $O(\log^4(n)/\alpha)$ time.
There are
$$O((n/\alpha)^2 \beta_{r+2}(\log (n)/\alpha)\log^2n\log(\log (n)/\alpha))$$
 tournament 
events, which are handled in a total of 
$$O((n/\alpha)^2 \beta_{r+2}(\log (n)/\alpha)\log^2n\log^2(\log (n)/\alpha))$$
processing time. The worst-case processing time of a
tournament event is $O(\log^2(\log (n)/\alpha))$. The data structure is also {\it local}, in the sense that each point
is stored, at any given time, at only $O(\log^2n/\alpha^2)$ places in the structure.
\end{theorem}
Concerning locality, we note that a point participates in $O(s)$ projection tournaments at each of $O(s\log^2n)$ tree nodes. If it wins in at least one of the projection tournaments at a node, it is fed to $O(s)$ 
directional tournaments. So it appears in $O(s^2\log n)$ places.

\noindent {\bf Remarks:} (1) Comparing this algorithm with the space-inefficient one of Section~\ref{sec:Prelim}, we note that they both use the 
same kind of tournaments, but here much fewer pairs of points 
($O^*(n/\alpha^2)$ instead of $O(n^2/\alpha)$) participate in the 
tournaments. The price we have to pay is that the test for an edge $pq$
to be stable is more involved. Moreover, keeping track of the subset of 
pairs that participate in the tournaments requires additional work,
which is facilitated by the range trees $\T_i$.

\medskip\noindent
(2) To be fair, we note that our $O^*(\cdot)$ notation hides polylogarithmic factors in $n$. Hence, comparing the analysis in this section with Theorem \ref{Thm:MaintainSDGPolyg}, we gain when $n$ is smaller than some threshold, which is exponential in $1/\alpha$.
\section{Properties of  SDG}\label{Sec:SDGProperties}
We conclude the paper by establishing some of the properties of stable Delaunay graphs. 

\paragraph{Near cocircularities do not show up in an SDG.}
Consider a critical event during the kinetic maintenance of the full
Delaunay triangulation, in which four points $a,b,c,d$ become cocircular,
in this order, along their circumcircle, with this circle being empty.
Just before the critical event, the Delaunay triangulation involved
two triangles, say, $abc$, $acd$. The Voronoi edge $e_{ac}$ shrinks
to a point (namely, to the circumcenter of $abcd$ at the critical event),
and, after the critical cocircularity, is replaced by the Voronoi edge
$e_{bd}$, which expands from the circumcenter as time progresses.

Our algorithm will detect the possibility of such an event before the criticality occurs,
when $e_{ac}$ becomes $\alpha$-short (or even before this happens). It will then remove this edge from the stable subgraph,
so the actual cocircularity will not be recorded. The new edge $e_{bd}$
will then be detected by the algorithm only when it becomes sufficiently long
(if this happens at all), and will then enter the stable Delaunay graph. In short,
critical cocircularities do not arise {\em at all} in our scheme.

As noted in the introduction, a Delaunay edge $ab$ (interior to the hull) is just about to become $\alpha$-short or $\alpha$-long when the sum of the opposite angles in its two adjacent Delaunay triangles is $\pi-\alpha$ (see Figure \ref{Fig:LongDelaunay}). This shows that changes in the stable Delaunay graph occur when the
``cocircularity defect'' of a nearly cocircular quadruple (i.e., the difference between $\pi$ and the sum of opposite angles in the quadrilateral spanned by the quadruple) is between
$\alpha$ and $c\alpha$, where $c$ is the constant used in our definitions in Section \ref{sec:ViaPolygonal} or Section \ref{Sec:ReduceS}.
Note that a degenerate case of cocircularity is a collinearity on the convex
hull.
Such collinearities also do not show up in the stable
Delaunay graph.\footnote{Even if they did show up, no real damage would be done, because the number of such collinearities is only $O^*(n^2)$; see, e.g., \cite{SA95}.} A hull collinearity between three nodes $a, b, c$ is
detected before it happens, when (or before) the corresponding Voronoi edge
becomes $\alpha$-short, in which case the angle $\angle acb$, where $c$ is the middle point of the (near-)collinearity becomes $\pi-\alpha$ (see Figure \ref{collinearity}). 
Therefore a hull edge is removed from the $\SDG$ if the Delaunay triangle is
almost collinear. The edge (or any new edge about to replace it) re-appears in the $\SDG$ when its corresponding
Voronoi edge is long enough, as before. 

\begin{figure}
\begin{center}
\input{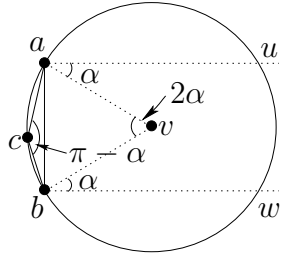} 
\caption{\small \sf The near collinearity that corresponds to a Voronoi edge
becoming $\alpha$-short.} \label{collinearity}
\end{center}
\end{figure}

\paragraph{SDGs are not too sparse.}
Consider the Voronoi cell $\Vor(p)$ of a point $p$, and suppose that $p$ has only one $\alpha$-long edge $e_{pq}$. Since the angle at which $p$ sees $e_{pq}$ is at most $\pi$, the sum of the angles at which $p$ sees the other edges is at least $\pi$, so $\Vor(p)$ has at least $\pi/\alpha$ $\alpha$-short edges. Let $m_1$ denote the number of points $p$ with this property. Then the sum of their degrees in $\DT(P)$ is at least $m_1(\pi/\alpha+1)$. Similarly, if $m_0$ points do not have any $\alpha$-long Voronoi edge, then the sum of their degrees is at least $2\pi m_0/\alpha$. Any other point at least two $\alpha$-long Voronoi edges and its degree is at least 3 if it is an interior point, or at least 2 otherwise. So the number of $\alpha$-long
edges is at least (recall that each $\alpha$-long edge is counted twice)
\begin{equation}
\label{Eq:long-edge}
n-m_1-m_0+m_1/2=n-(m_1+2m_0)/2 .
\end{equation}
Let $h$ denote the number of hull vertices. Since the sum of the degrees is $6n-2h-6$, we get

$$3(n-h-m_1-m_0)+2h+m_1\left(\frac{\pi}{\alpha}+1\right)+2m_0\frac{\pi}{\alpha}
\leq 6n-2h-6,$$
implying that
$$m_1+2m_0\leq \frac{3n}{\pi/\alpha-2}.$$
Plugging this inequality in (\ref{Eq:long-edge}), we conclude that the number of $\alpha$-long edges is at least
$$ n\left[1-\frac{3}{2(\pi/\alpha-2)}\right].$$
As $\alpha$ decreases, the number of edges in the SDG is always at 
least a quantity that gets closer to $n$.
This is nearly tight, since there exist $n$-point sets for which the number of stable edges is only roughly $n$, see Figure \ref{Fig:ShiftedGrid}.

\begin{figure}
\begin{center}
\input{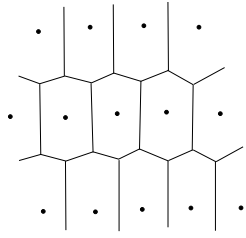}
\caption{\small \sf If the points of $P$ lie on a sufficiently spaced shifted grid then the number of $\alpha$-long edges in $\VD(P)$ (the vertical ones) is close to $n$.} \label{Fig:ShiftedGrid}
\end{center}
\end{figure}

\paragraph{Closest pairs, crusts, $\beta$-skeleta, and the SDG.}
Let $\beta\geq 1$, and let $P$ be a set of $n$ points in the plane. 
The \textit{$\beta$-skeleton} of $P$ is a graph on $P$ that 
consists of all the edges $pq$ such that the union of the two disks of 
radius $(\beta/2)d(p,q)$, touching $p$ and $q$, does not contain any 
point of $P\setminus\{p,q\}$. See, e.g., \cite{Crusts,Skeletons} for 
properties of the $\beta$-skeleton, and for its applications in surface reconstruction. 
We show that the edges of the $\beta$-skeleton are $\alpha$-stable 
in $\DT(P)$, provided $\beta\geq 1+\Omega(\alpha^2)$.
In Figure \ref{Fig:Skeleton} we sketch a straightforward proof of the fact that the edges of the $\beta$-skeleton are $\alpha$-stable in $\DT(P)$, provided that $\beta\geq 1+\Omega(\alpha^2)$.

\begin{figure}[htbp]
\begin{center}
\input{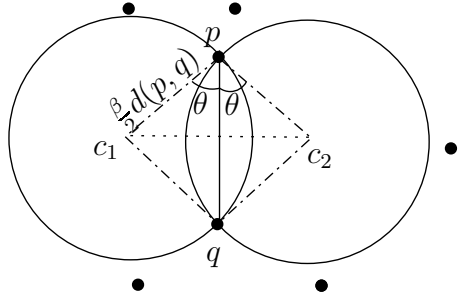}
\caption {\small \sf An edge $pq$ of the $\beta$-skeleton of $P$ (for $\beta>1$). $c_1$ and $c_2$ are centers of the two $P$-empty disks of radius $(\beta/2)d(p,q)$ touching $p$ and $q$. Clearly, each of $p,q$ sees the Voronoi edge $e_{pq}$ at an angle at least $2\theta=\angle c_1pq+\angle c_2pq$ (so it is $2\theta$-stable). We have $1/\beta=\cos \theta\approx 1-\theta^2/2$ or $\beta=1+\Theta(\theta^2)$. That is, for $\beta\geq 1+\Omega(\alpha^2)$ every edge of the $\beta$-skeleton is $\alpha$-stable.}\label{Fig:Skeleton}
\end{center}
\end{figure}

 A similar argument shows that the stable Delaunay graph contains the 
closest pair in $P(t)$ as well as the crust of a set of points sampled
sufficiently densely along a 1-dimensional curve (see \cite{Amenta,Crusts} for the definition of crusts and their applications in surface 
reconstruction). 
We only sketch the argument for closest pairs: If $(p,q)$ is a closest pair then $pq\in \DT(P)$, and the two adjacent Delaunay triangles $\triangle pqr^+,\triangle pqr^-$ are such that their angles of $r^+,r^-$ are at most $\pi/3$ each, so $e_{pq}$ is $(\pi/3)$-long, ensuring that $pq$ belongs to any stable subgraph for $\alpha$ sufficiently small; see \cite{KineticNeighbors} for more details. 
We omit the proof for crusts, which is fairly straightforward.

\begin{figure}
\begin{center}
\input{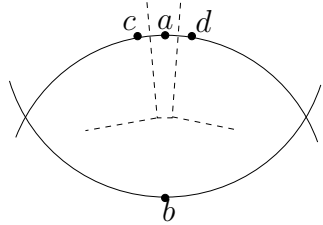} 
\caption{\small \sf $ab$ is an edge of the relative neighborhood graph but not of
$\SDG$.}
\label{norng1}
\end{center}
\end{figure}

In contrast, stable Delaunay graphs need not contain all the
edges of several other important subgraphs of the Delaunay
triangulation, including the Euclidean minimum spanning tree, the
Gabriel graph, the relative neighborhood graph, and the
all-nearest-neighbors graph. An illustration for the relative neighborhood graph is given in Figure \ref{norng1}. As a matter of fact, the stable
Delaunay graph need not even be connected, as is illustrated in
Figure~\ref{norng2}.

\begin{figure}
\begin{center}
\input{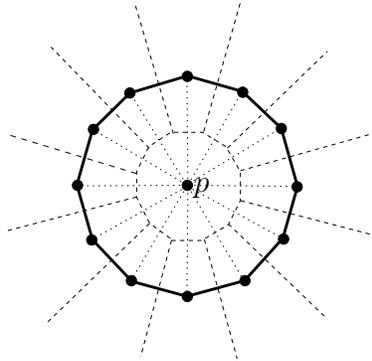}
\caption{\small \sf A wheel-like configuration that disconnects $p$ in the
stable Delaunay graph. The Voronoi diagram is drawn with dashed
lines, the stable Delaunay edges are drawn as solid, and the remaining
Delaunay edges as dotted edges. The points of the ``wheel" need not be cocircular.}
\label{norng2}
\end{center}
\end{figure}

\paragraph{Completing SDG into a triangulation.}
As argued above, the Delaunay edges that are missing in the stable
subgraph correspond to nearly cocircular quadruples of points, or
to nearly collinear triples of points near the boundary of the convex
hull. Arguably, these missing edges carry little information, because
they may ``flicker" in and out of the Delaunay triangulation even when the points
move just slightly (so that all angles determined by the triples of points change only slightly). Nevertheless, in many applications it is desirable
(or essential) to complete the stable subgraph into {\em some} triangulation,
preferrably one that is also stable in the combinatorial sense---it undergoes
only nearly quadratically many topological changes.

By the analysis in Section \ref{Sec:polygProp} we can achieve part of this goal by maintaining the full Delaunay triangulation $\DT^\poly(P)$ under the polygonal norm induced by the regular $k$-gon $Q_k$. This diagram experiences only a nearly quadratic number of topological changes, is easy to maintain, and contains all the stable Euclidean Delaunay edges, for an appropriate choice of $k\approx 1/\alpha$. Moreover, the union of its triangles is simply connected --- it has no holes. Unfortunately, in general it is not a triangulation of the entire convex hull of $P$, as illustrated in Figure \ref{Fig:AlmostTriangulation}.

\begin{figure}
\begin{center}
\input{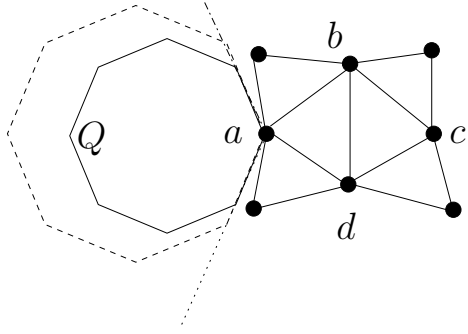}
\caption{\small \sf The triangulation $\DT^\poly(P)$ of an 8-point set $P$.
The points $a,b,c,d$, which do not lie on the convex hull of $P$, still lie on the boundary of the union of the triangles of $\DT^\poly(P)$
because, for each of these points we can place an arbitrary large homothetic interior-empty copy of $Q$ which touches that point.}
\label{Fig:AlmostTriangulation}
\end{center}
\end{figure}

For the time being, we leave it as an open problem to come up with a 
simple and ``stable" scheme for filling the gaps between the triangles 
of $\DT^\poly(P)$ and the edges of the convex hull. 
It might be possible to extend the kinetic triangulation scheme 
developed in \cite{KRS}, so as to kinetically maintain a triangulation of the
``fringes" between $\DT^\poly(P)$ and the convex hull of $P$, which is simple to define, easy to maintain, and undergoes only nearly quadratically many topological changes.

Of course, if we only want to maintain a triangulation of $P$ that experiences only a nearly quadratically many topological changes, then we can use the scheme in \cite{KRS}, or the earlier, somewhat more involved scheme in \cite{AWY}. However, if we want to keep the triangulation ``as Delaunay as possible", we should include in it the stable portion of $\DT$, and then the efficient completion of it, as mentioned above, becomes an issue, not yet resolved.

\paragraph{Nearly Euclidean norms and some of their properties.}
One way of interpreting the results of Section 3 is that the stability of Delaunay edges is preserved, in an appropriately defined sense, if we replace the Euclidean norm by the polygonal norm induced by the regular $k$-gon $Q_k$ (for $k\approx 1/\alpha$). That is, stable edges in one Delaunay triangulation are also edges of the other triangulation, and are stable there too. Here we note that there is nothing special about $Q_k$: The same property holds if we replace the Euclidean norm by any sufficiently close norm (or convex distance function \cite{CD}).

Specifically, let $Q$ be a closed convex set in the plane that is contained in the
unit disk $D_0$ and contains the disk $D'_0 = (\cos\alpha) D_0$ that
is concentric with $D_0$ and scaled by the factor $\cos\alpha$.
This
is equivalent to requiring that the Hausdorff distance $H(Q,D_0)$ 
between $Q$ and $D_0$ be at most $1-\cos\alpha\approx \alpha^2/2$. 
We define the center of $Q$ to coincide with the common center of 
$D_0$ and $D'_0$.


$Q$ induces a convex distance function $d_Q$, defined by $d_Q(x,y)=\min \{\lambda\mid y\in x+\lambda Q\}$. Consider the Voronoi diagram $\Vor^Q(P)$
of $P$ induced by $d_Q$, and the corresponding Delaunay triangulation $\DT^Q(P)$. We omit here the detailed analysis of the structure of these diagrams, which is similar to that for the norm induced by $Q_k$, as presented in Section \ref{Sec:polygProp}. See also \cite{Chew,CD} for more details. Call an edge $e_{pq}$ of $\Vor^Q(P)$ $\alpha$-stable if the following property holds: Let $u$ and $v$ be the endpoints of $e_{pq}$, and let $Q_u,Q_v$ be the two homothetic copies of $Q$ that are centered at $u,v$, respectively, and touch $p$ and $q$. Then we require that the angle between the 
supporting lines at $p$
(for simplicity, assume that $Q$ is smooth, and so has a unique supporting line at $p$ (and at $q$); otherwise, the condition should hold for any pair of supporting lines at $p$ or at $q$) 
to $Q_u$ and $Q_v$ is at least $\alpha$, and that the same holds at $q$.
In this case we refer to the edge $pq$ of $\DT^Q(P)$ as $\alpha$-stable.

Note that $Q_k$-stability was (implicitly) defined in a different manner in Section \ref{Sec:polygProp}, based on the number of breakpoints of the corresponding Voronoi edges. Nevertheless, it is easy to verify that the two definitions are essentially identical.
\begin{figure}[hbt]
\begin{center}
\input{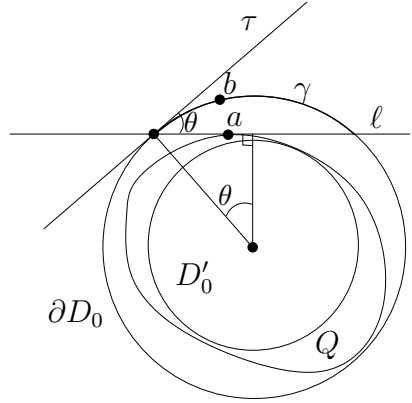}
\caption{\small \sf An Illustration for Claim \ref{Q1}. 
 \label{fig:norm1}}
\end{center}
\end{figure}

A useful property of such a set $Q$ is the following: 
\begin{claim} \label{Q1}
Let $a$ be a point on $\bd Q$ and let
$\ell$ be a supporting line to $Q$ at $a$.
Let $b$ be the point on $\bd D_0$
closest to $a$ ($a$ and $b$ lie on the same radius from the center
$o$).
Let $\gamma$ be the arc of  $\bd D_0$, containing $b$, and bounded by
the intersection points of $\ell$ with  $\bd D_0$.
 Then the angle between $\ell$ and  the tangent, $\tau$, to $D_0$ at any point
along $\gamma$, 
 is at most $\alpha$. 
\end{claim}

\begin{proof}
Denote this angle by $\theta$.
Clearly $\theta$ is maximized when $\tau$ is tangent to $D_0$
at an intersection of $\ell$ and $\bd D_0$.
See Figure \ref{fig:norm1}.
It is easy to verify that
the distance from 
$o$ to $\ell$ is $\cos\theta$. But this distance has to be at least
$\cos\alpha$, or else $\bd Q$ would have contained a point inside 
$D'_0$, contrary to assumption. Hence we have 
$\cos\theta > \cos\alpha$, and thus $\theta < \alpha$, as claimed.
\end{proof}

We need a few more properties:

\begin{figure}[hbt]
\begin{center}
\input{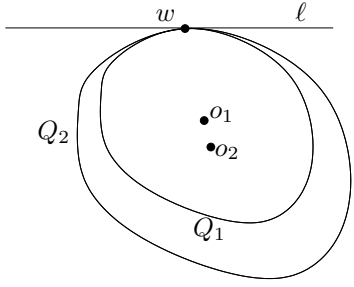}
\caption{\small \sf An Illustration for Claim \ref{Q2}. 
 \label{fig:norm2}}
\end{center}
\end{figure}

\begin{claim} \label{Q2}
Let $Q_1$ and $Q_2$ be two homothetic copies of $Q$ and let $w$ be a
point such that (i) $w$ lies on $\bd Q_1$ and on $\bd Q_2$, and
(ii) $w$ and the respective centers $o_1$, $o_2$ of $Q_1$, $Q_2$
are collinear. Then $Q_1$ and $Q_2$ are tangent to each other at $w$;
more precisely, they have a common supporting line at $w$, and, assuming $\partial Q$ to be smooth, $w$ is the only point of intersection of $\partial Q_1\cap \partial Q_2$ (otherwise, $\partial Q_1\cap \partial Q_2$ is a single connected arc containing $w$.).
\end{claim}
\begin{proof}
Map each of $Q_1$, $Q_2$ back to the standard placement of $Q$, by
translation and scaling, and note that both transformations map $w$
to the same point $w_0$ on $\bd Q$. Let $\ell_0$ be a supporting line
of $Q$ at $w_0$, and let $\ell_1$, $\ell_2$ be the forward images of
$\ell$ under the mappings of $Q$ to $Q_1$ and to $Q_2$, respectively.
Clearly, $\ell_1$ and $\ell_2$ coincide, and are a common supporting
line of $Q_1$ and $Q_2$ at $w$.
See Figure \ref{fig:norm2}. The other asserted property follows immediately if $\partial Q$ is smooth, and can easily be shown to hold in the non-smooth case too; we omit the routine argument.
\end{proof}

\begin{claim} \label{Q3}
Let $a$ and $b$ be two points on $\bd Q$, and let $\ell_a$ and $\ell_b$
be supporting lines of $Q$ at $a$ and $b$, respectively. Then the
difference between the angles that $\ell_a$ and $\ell_b$ form with
$ab$ is at most $2\alpha$.
\end{claim}

\begin{proof}
Denote the two angles in the claim by $\theta_a$ and $\theta_b$,
respectively.
Let $a'$ (resp., $b'$) be the point on $\bd D_0$ nearest to (and
co-radial with) $a$ (resp., $b$). Let $\tau_1$, $\tau_2$ denote the
respective tangents to $D_0$ at $a'$ and at $b'$.  Clearly, the 
respective angles $\theta_1$, $\theta_2$ between the chord $a'b'$ 
of $D_0$ and $\tau_1$, $\tau_2$ are equal. By Claim~\ref{Q1}, we 
have $|\theta_1-\theta_a|\le\alpha$ and
$|\theta_2-\theta_b|\le\alpha$, and the claim follows.
\end{proof}

\paragraph{The connection between Euclidean stability and $Q$-stability.} 
Let $e_{pq}$ be a $t\alpha$-long Voronoi edge of the Euclidean diagram, for $t\ge 9$, 
and let $u,v$ denote its endpoints. 
Let $D_u$ and $D_v$ denote the disks centered respectively
at $u,v$, whose boundaries pass through $p$ and $q$, and let $D$ be a
disk whose boundary passes through $p$ and $q$, so that 
$D\subset D_u\cup D_v$ and the angles between the tangents to $D$ and
to $D_u$ and $D_v$ at $p$ (or at $q$) are at least $m\alpha$ each, where
$m \geq 4$. (Recall that the angle between the tangents to $D_u$ and $D_v$ us at least $t\alpha\geq 9\alpha$.)

\begin{figure}[hbt]
\begin{center}
\input{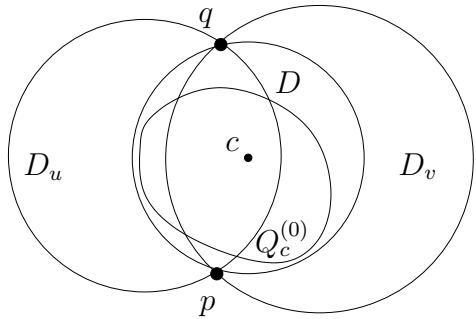}
\caption{\small \sf The homothetic copy $Q^{(0)}_c$.
 \label{fig:norm3}}
\end{center}
\end{figure}

Let $c$ and $\rho$ denote  the center and radius of $D$, respectively.
Note that $c$ lies on $e_{pq}$ ``somewhere in the middle'', because of
the angle condition assumed above.
Let $Q^{(0)}_c$ denote the homothetic copy of $Q$ centered at $c$ and
scaled by $\rho$, so $Q^{(0)}_c$ is fully contained in $D$ and thus
also in $D_u\cup D_v$, implying that $Q^{(0)}_c$ is {\em empty}---it
does not contain any point of $P$ in its interior. (This scaling makes the
unit circle $D_0$ bounding $Q$ coincide with $D$.) See Figure \ref{fig:norm3}.

Expand $Q^{(0)}_c$ about its center $c$ until the first time it 
touches either $p$ or $q$. Suppose, without loss of generality, 
that it touches $p$. Denote this placement of $Q$ as $Q_c$.
Let $\ell_p$ denote a supporting line of $Q_c$ at $p$. We claim that the angle between $\ell_p$ and the tangent $\tau_p$ to $D$ at $p$ is at most $\alpha$. Indeed, let $\ell_p^-,\ell_p^+$ denote the tangents from $p$ to $Q_c^{(0)}$. By Claim \ref{Q1}, the angles that they form with the tangent $\tau_p$ to $D$ at $p$ are at most $\alpha$ each. As $Q_c^{(0)}$ is expanded to $Q_c$, these tangents rotate towards each other, one clockwise and one counterclockwise so when they coincide (at $Q_0$) the resulting supporting line $\ell_p$ lies inside the double wedge between them. Since $\tau_p$ also lies inside this double wedge, and forms an angle of at most $\alpha$ with each of them, it follows that $\ell_p$ must form an angle of at most $\alpha$ with $\tau_p$, as claimed.

Since the angle between the tangent $\tau_p$ to $D$ at $p$ and the tangent
$\tau_p^v$
to $D_v$ at $p$ is at least $m\alpha$ it follows that the angle between
$\ell_p$ and $\tau_p^v$ is at least $(m-1) \alpha$. 
A similar argument shows that the angle between $\ell_p$ and 
the tangent $\tau_p^u$ to $D_u$ at $p$ is at least  $(m-1) \alpha$. 
 
\begin{figure}[hbt]
\begin{center}
\input{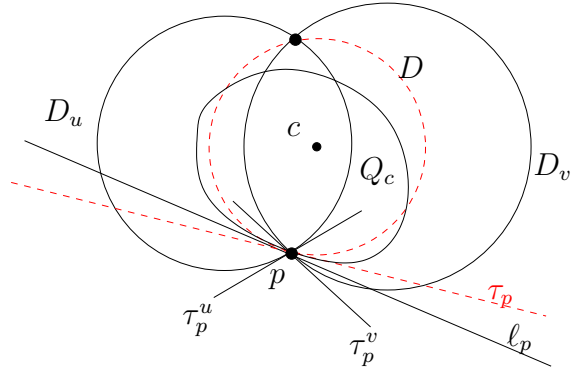}
\caption{\small \sf The homothetic copy $Q_c$.
 \label{fig:norm4}}
\end{center}
\end{figure}

Now expand $Q_c$ by moving its center along the line passing
through $p$ and $c$, away from $p$, and scale it appropriately so 
that its boundary continues to pass through $p$, until it touches 
$q$ too. Denote the center of the new placement as $c'$, and
the placement itself as $Q_{c'}$. Let $D_{c'}$ be the 
corresponding homothetic copy of $D_0$ centered at $c'$ and bounding
$Q_{c'}$. See Figure \ref{fig:norm4}.

\begin{figure}[hbt]
\begin{center}
\input{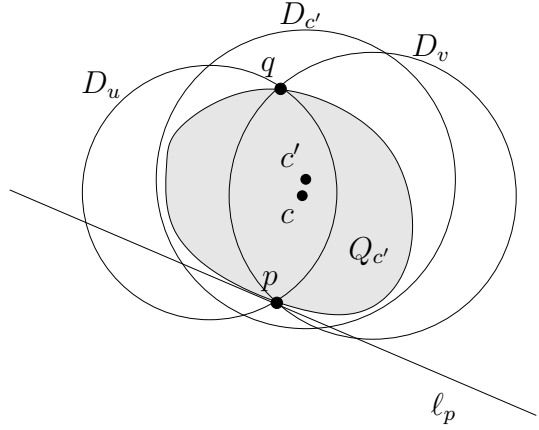}
\caption{\small \sf The homothetic copy $Q_{c'}$.
 \label{fig:norm5}}
\end{center}
\end{figure}

We argue that  $Q_{c'}$ is empty.
 By Claim~\ref{Q2}, $\ell_p$ is also
a supporting line of $Q_{c'}$ at $p$.
Refer to Figure \ref{fig:norm6}.
We denote by $x_p$ and $y_p$ the intersections of the supporting line
$\ell_p$ with $\partial D_{c'}$ and $\partial D_v$, respectively.
 We denote by $z$ the intersection of $\partial D_{c'}$ and $\partial D_v$ that lies on the same side of $\ell_p$ as $q$.
The angle $\angle pzx_p$ is at most $\alpha$ since by Claim \ref{Q1}
the angle between $\ell_p$ and the tangent to $D_{c'}$ at $x_p$ is at most 
$\alpha$.
On the other hand the
 angle $\angle pzy_p$ is at least $(m-1)\alpha$ since the angle
between $\ell_p$ and $\tau_p^v$ at $p$ is at least 
$(m-1)\alpha$. So it follows that
the segment $px_p$ is fully contained in $D_v$.
Since the ray $\overline{zp}$ meets $\partial D_v$ (at $p$) before meeting $\partial D_{c'}$, and the ray $\overline{zx_p}$ meets $\partial D_{c'}$ (at $x_p$) before meeting $\partial D_v$, it follows that $\partial D_{c'}$ and $\partial D_v$ intersect at a point on a ray between $\overline{zp}$ and $\overline{zx_p}$.

\begin{figure}[hbt]
\begin{center}
\input{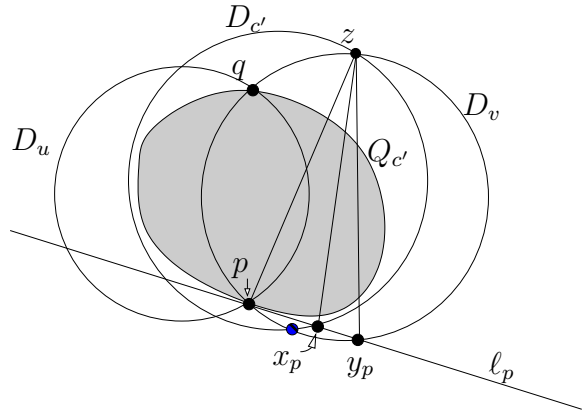}
\caption{\small \sf The segment $px_p$ is fully contained in $D_v$. The circles
$\partial D_{c'}, \partial D_v$ intersect at a point on a ray emanating from $z$ between $zp$ and $z_{x_p}$.
 \label{fig:norm6}}
\end{center}
\end{figure}

Let $\ell_q$ denote a 
supporting line of $Q_{c'}$ at $q$. By Claim~\ref{Q3}, the angles
between $pq$ and the lines $\ell_p$, $\ell_q$ differ by at most
$2\alpha$.
Since each of the angles between $\ell_p$ and 
the two tangents
$\tau_p^v$ 
and $\tau_p^u$ is at least $(m-1) \alpha$, it follows that
each of the angles between $\ell_q$ and the two 
tangents $\tau_q^u$ and  $\tau_q^v$ to $D_u$ and $D_v$,
respectively, at $q$, is at least $(m-3)\alpha$.

Refer now to Figure \ref{fig:norm7}.
We denote by $z'$ the intersection of $D_{c'}$ and $D_v$ distinct from $z$,
and we denote by $x_q,y_q$ the intersections between $\ell_q$ and
 $D_{c'},D_v$, respectively. An argument analogous to the one given
before shows that $\angle qz'x_q \le \alpha$ while 
$\angle qz'y_q \ge  (m-3)\alpha$. It follows that the segment 
$qx_q$ is fully contained in $D_v$ and we have an intersection between 
$\partial D_{c'}$ and $\partial D_v$ on a ray emanating from $z'$ between the ray from $z'$ to
$q$ and the ray from $z'$ to $x_q$.

\begin{figure}[hbt]
\begin{center}
\input{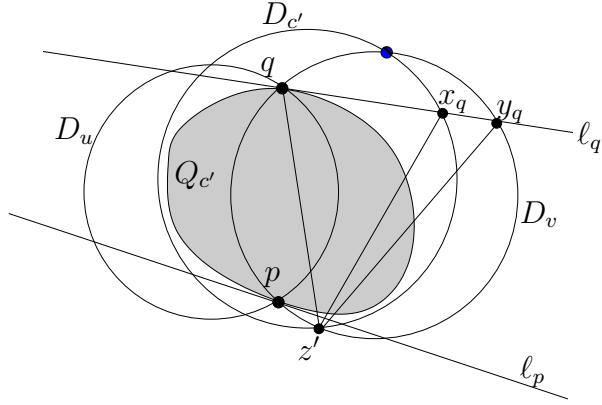}
\caption{\small \sf The segment $qx_q$ is fully contained in $D_v$. The circles
$\partial D_{c'}, \partial D_v$ intersect at a point on a ray emanating from $z'$ between $z'q$ and $z'x_q$.
 \label{fig:norm7}}
\end{center}
\end{figure}

Our argument about the position of the intersections between
$D_{c'}$ and $D_v$ implies that the entire section of 
$\partial D_{c'}$ between $x_p$ and $x_q$ is contained $D_v$. Therefore
the portion of $Q_{c'}$ to the right of the line through $p$ and $q$ (in the configuration depicted in the figures) is fully contained in $D_v$.
A symmetric argument shows that the portion of 
$Q_{c'}$ to the left of the line 
through $p$ and $q$ is fully contained in $D_u$. Since $D_u \cup D_v$ is empty we conclude that
$Q_{c'}$ is empty.

 The emptiness of $Q_{c'}$ implies that $p$ and 
$q$ are neighbors in the $Q$-Voronoi diagram, and that $c'$ lies on 
their common $Q$-Voronoi edge $e^Q_{pq}$.

We thus obtain the following theorem.
\begin{theorem} 
Let $P$, $\alpha$, and $Q$ be as above. Then (i) every $9\alpha$-stable edge of the Euclidean Delaunay triangulation is an $\alpha$-stable edge of $\DT^Q(P)$. (ii) Conversely, every $9\alpha$-stable edge of $\DT^Q(P)$ is also an $\alpha$-stable edge in the Euclidean norm.
\end{theorem}
Note that parts (i) and (ii) are generalizations of Lemmas \ref{Thm:LongEucPoly} and \ref{Thm:LongPolygEuc}, respectively (with weaker constants).
\begin{proof}
Part (i) follows directly from the preceding analysis. Indeed, let $pq$ be a $t\alpha$-stable Delaunay edge, for $t\geq 9$, whose Voronoi counterpart has endpoints $u$ and $v$. Let $Q_{c'}$ be the homothetic placement of $Q$, with center $c'$, that touches $p$ and $q$. We have shown that $Q_{c'}$ has empty interior if the ray $\rho=\overline{pc'}$ lies between $\overline{pu}$ and $\overline{pv}$ and spans an angle of at least $4\alpha$ with each of them. Assuming $t\geq 9$, such rays $\rho$ form a cone of size $(t-8)\alpha>\alpha$, which, in turn, gives the first part of the theorem. 

Part (ii) follows from part (i) by repeating, almost verbatim, the proof of Lemma \ref{Thm:LongPolygEuc}.
\end{proof}

There are many interesting open problems that arise here. One of the main problems is to extend the class of sets $Q$ for which a near quadratic bound on the number of topological changes in $\DT^Q(P)$, under algebraic motion of bounded degree of the points of $P$, can be established.


\paragraph{Acknowledgements.}
{\small Pankaj Agarwal was supported by NSF under grants CNS-05-40347, CCF-06 -35000, CCF-09-40671
               and DEB-04-25465, by ARO grants
               W911NF-07-1-0376 and W911NF-08-1-0452, by an
               NIH grant 1P50-GM-08183-01, by a DOE grant
               OEG-P200A070505, and by Grant 2006/194 from the
                U.S.--Israel Binational Science Foundation.
    Leo Guibas was supported by NSF grants CCR-0204486,
    ITR-0086013, ITR-0205671, ARO grant DAAD19-03-1-0331, as well as by
    the Bio-X consortium at Stanford.
Haim Kaplan was partially supported by Grant 2006/204 from the U.S.--Israel
Binational Science Foundation, project number 2006204, and by Grants 975/06 and 822/10 from
the 
Israel Science Fund.
 Micha Sharir was supported by NSF Grants CCF-05-14079 and CCF-08-30272, 
    by Grant 338/09 from the Israel Science Fund,
    and by the Hermann Minkowski--MINERVA Center for Geometry at Tel
    Aviv University. Natan Rubin was supported by Grants 975/06 and 338/09 from the Israel Science Fund.}


\begin{thebibliography}{10}





\bibitem{KineticNeighbors}
P. K. Agarwal, H. Kaplan and M. Sharir, Kinetic and dynamic data
structures for closest pair and all nearest neighbors, \emph{ACM
Trans. Algorithms} 5 (1) (2008), Art.~4.

\bibitem{AWY}
P. K. Agarwal, Y. Wang and H. Yu,
A 2D kinetic triangulation with near-quadratic topological changes,
\textit{Discrete Comput. Geom.} 36 (2006), 573--592.


\bibitem{Amenta}
N. Amenta and M. Bern, 
Surface reconstruction by Voronoi filtering,
{\em Discrete Comput. Geom.}, 22 (1999), 481--504.

\bibitem{Crusts}
N. Amenta, M. W. Bern and D. Eppstein, The crust and beta-skeleton: combinatorial curve reconstruction, 
{\it Graphic. Models and Image Processing} 60 (2) (1998), 125--135. 



\bibitem{AK}
F. Aurenhammer and R. Klein,
Voronoi diagrams,
in {\it Handbook of Computational Geometry},
J.-R. Sack and J. Urrutia, Eds.,
Elsevier, Amsterdam, 2000,
pages 201--290.

\bibitem{bgh-dsmd-99}
J.~Basch, L.~J. Guibas and J.~Hershberger,
Data structures for mobile data,
{\em J. Algorithms} 31 (1) (1999), 1--28.

\bibitem{Chew}
L. P. Chew,
Near-quadratic bounds for the $L_1$ Voronoi diagram of moving points,
{\em Comput. Geom. Theory Appl.}  7 (1997), 73--80.

\bibitem{CD}
L. P. Chew and R. L. Drysdale,
Voronoi diagrams based on convex distance functions,
{\em Proc. First Annu. ACM Sympos. Comput. Geom.}, 1985, pp.~235--244.


\bibitem{d-slsv-34}
B.~Delaunay,
Sur la sph{\`e}re vide. {A} la memoire de {Georges} {Voronoi},
{\em Izv. Akad. Nauk SSSR, Otdelenie Matematicheskih i Estestvennyh
Nauk} 7 (1934), 793--800.

\bibitem{TOPP}
E.~D.~Demaine, J.~S.~B.~Mitchell, and J.~O'Rourke,\\
The Open Problems Project,
\texttt{http://www.cs.smith.edu/\~{ }orourke/TOPP/}.

\bibitem{Ed2}
H. Edelsbrunner,
{\em Geometry and Topology for Mesh Generation},
Cambridge University Press, Cambride, 2001.

\bibitem{285869}
L.~J. Guibas,
Modeling motion,
In J.~E. Goodman and J.~O'Rourke, editors, {\em Handbook of Discrete
and Computational Geometry}. CRC Press, Inc., Boca Raton, FL, USA, second
edition, 2004, pages 1117--1134.

\bibitem{g-kdssar-98}
L.~J. Guibas,
Kinetic data structures --- a state of the art report,
In P.~K. Agarwal, L.~E. Kavraki and M.~Mason, editors, {\em Proc.
Workshop Algorithmic Found. Robot.}, pages 191--209. A. K. Peters, Wellesley,
MA, 1998.

\bibitem{gmr-vdmpp-92}
L.~J. Guibas, J.~S.~B. Mitchell and T.~Roos,
Voronoi diagrams of moving points in the plane,
{\em Proc. 17th Internat. Workshop Graph-Theoret. Concepts Comput.
Sci.}, volume 570 of {\em Lecture Notes Comput. Sci.}, pages 113--125.
Springer-Verlag, 1992.

\bibitem{IKLM}
C. Icking, R. Klein, N.-M. L\^{e} and L. Ma,
Convex distance functions in 3-space are different, {\it Fundam. Inform.} 22 (4) (1995), 331--352.

\bibitem{KRS}
H. Kaplan, N. Rubin and M. Sharir, 
A kinetic triangulation scheme for moving points in the plane, {\it Comput. Geom. Theory Appl.} 44 (2011), 191--205.

\bibitem{KLPS}
K. Kedem, R. Livne, J. Pach, and M. Sharir, On the union of Jordan regions and collision-free translational motion amidst polygonal obstacles, {\em Discrete Comput. Geom.} 1 (1986), 59--70.

\bibitem{Skeletons}
D. Kirkpatrick and J. D. Radke, A framework for computational morphology, \textit{Computational Geometry} (G. Toussaint, ed.), North-Holland (1985), 217--248.

\bibitem{LS}
D. Leven and M. Sharir,
Planning a purely translational motion for a convex object in
two--dimensional space using generalized Voronoi diagrams,
{\it Discrete Comput. Geom.} 2 (1987), 9--31.

\bibitem{Mehlhorn} 
K. Mehlhorn, \textit{Data Structures and Algorithms 1: Sorting and Searching}, Springer Verlag, Berlin 1984.

\bibitem{NR73} 
J. Nievergelt and E. M. Reingold, Binary search trees of bounded balance, {\it SIAM J. Comput.} 2 (1973),
33--43.

\bibitem{SA95}
M.~Sharir and P.~K. Agarwal,
{\em Davenport-Schinzel Sequences and Their Geometric Applications},
Cambridge University Press, New York, 1995.

\end{thebibliography}
\end{document}